\documentclass{aastex}
\usepackage{spr-astr-addons}
\usepackage{url}

\newtheorem{theorem}{Theorem}

\newtheorem{corollary}[theorem]{Corollary}

\newtheorem{lemma}[theorem]{Lemma}

\newenvironment{proof}[1][Proof]{\noindent\textbf{#1.} }{\ \rule{0.5em}{0.5em}}

\RequirePackage{color}

\begin{document}

\title{On the planar central configurations of rhomboidal and triangular four- and five-body problems}
\slugcomment{Not to appear in Nonlearned J., 45.}
\shorttitle{Rhomboidal and triangular four- and five-body problems}
\shortauthors{Shoaib et al.}

\author{M. Shoaib\altaffilmark{1}}
\affil{Abu Dhabi Men's College, Higher Colleges of Technology, P.O. Box 25035, Abu Dhabi, United Arab Emirates}
\and
\author{A. R. Kashif\altaffilmark{2}}
\affil{Department of Mathematics, Capital University of Sciences and Technology, Islamabad, Pakistan}
\and
\author{I. Sz\"{u}cs-Csillik\altaffilmark{3}}
\affil{Astronomical Institute of the Romanian Academy, Cire\c{s}ilor 19, 400487 Cluj-Napoca, Romania}

\begin{abstract}
We consider a symmetric five-body problem with three unequal collinear masses on the axis of symmetry.
The remaining two masses are symmetrically placed on both sides of the axis of symmetry.
Regions of possible central configurations are derived for the four- and five-body problems.
These regions are determined analytically and explored numerically.
The equations of motion are regularized using Levi-Civita type transformations and then the phase space is investigated for chaotic and periodic orbits by means of Poincar\'{e} surface of sections.
\end{abstract}

\section{Introduction}

Studying the four- and five-body problems are very important as (for example) it is known that approximately two-third of the stars in our Galaxy are part of
multi-stellar systems.
This paper focuses on Central Configurations (CC), a very important topic in  the gravitational $n$-body problems (cf.
\cite{DengZhang2014}, \cite{MacMillan}, \citep{simo} and \cite{Libre2008}). 
A central configuration, subject of this paper, in the $n$-body problem is a particular position of the $n$ particles, where the position and acceleration vectors are proportional with the same constant of proportionality (see \cite{meyerhall1992} and \cite{Perez2007}).
CC are useful in understanding the nature of solutions near singularities. They are also useful in providing explicit homographic solutions of the equations of motion and families of periodic solutions.
 

In this paper, we explore the central configurations of rhomboidal and triangular four- and five-body problems.
In 2012, Bakker and Simmons gave a linear stability analysis of a rhomboidal four-body
problem and have shown that isolated binary collisions are regularizable at
the origin. \cite{Lacomba1992} had earlier studied the same problem to
regularize binary collisions. In a 1993 paper, Lacomba and Perez-Chavela
studied it's escape and capture orbits.
Furthermore, \cite{Yan2012} studied the existence and linear stability of periodic orbits
for equal masses and \cite{JI2000} used the Poincar\'{e} sections to find regions of stability
for the rhomboidal four-body problem.  \cite{Corbera2016} show that if the diagonals of a four-body convex central configuration are perpendicular, then the configuration must be a rhombus. We study a similar problem in theorem 3 in which we explicitly show the regions where such a configuration exists. Other studies on the rhomboidal four-body problem include
\cite{Chen2001},  \cite{Delgado1991}, \cite{Hampton2011} and \cite{Waldvogel2012}.

 \cite{Ollongren}  studies a restricted five-body problem having
three bodies of equal mass, $m$, placed on the vertices of the equilateral
triangle; they revolve in the plane of the triangle around their
gravitational centre in circular orbits under the influence of their mutual
gravitational attraction; at the centre a mass of $\beta m$ is present where
$\beta \geq 0$. A fifth body of negligible mass compared to $m$ moves in the
plane under the gravitational attraction of the other bodies. They discuss the
existence and location of the Lagrangian equilibrium points and show that
there are 9 Lagrangian equilibrium points.

\cite{Roberts}  discusses  the relative equilibria for a special case of the 5-body
problem. He considers a configuration which consists of four bodies at the
vertices of a rhombus, with opposite vertices having the same mass, and a
central body. He shows that there exist a one parameter family of degenerate
relative equilibria where the four equal masses are positioned at the
vertices of a rhombus with the remaining body located at the centre.
 \cite{MiocBlaga} discuss the same problem
but in the post Newtonian field of Manev. They prove the existence of
mono-parametric families of relative equilibria for the masses $%
(m_{0},1,m,1,m)$, where $m_{0}$ is the central mass, and prove that the
Manev five-body problem with masses $(m_{0},1,m,1,m)$ admits relative
equilibria regardless of the value of the mass of the central body. A
continuum of such equilibria (as in the Newtonian field) does not exist in
the Manev rhomboidal five-body problem.
\cite{Corbera2006} found infinitely many periodic orbits for a rhomboidal
five-body problem when two pairs of masses are placed at the vertices of a
rhombus and a stationary fifth mass is placed at the origin.
Moreover, \cite{Shoaib2012} derive regions of central configurations for the same
model. \cite{Kulesza2011} and \cite{Marchesin2013} study various aspects of
restricted rhomboidal five-body problem with four positive masses on
the vertices of a rhombus and the fifth infinitesimal mass in the plane of the
four masses. \cite{Martha2016} study the central configurations of a rhombus like five-body problem with four equal masses and a fifth infinitesimal mass. \cite{Shoaib2013} and \cite{Marian2010} derive regions of central configuration for a special case of rhomboidal five-body problem with two axes of symmetry and two pairs of equal masses and fifth mass at the centre of mass. They show the existence of regions with a continuous family of central configurations. In this paper, in theorem 3, we show that if the fifth mass is moved from the centre of mass to anywhere on the axes of symmetry then there are no central configurations. Other studies on the rhomboidal and triangular five-body problem include \cite{Lee2009}, \cite{Shoaib2016} and \cite{Hampton2005}

In the present paper we set up a five-body problem which has three collinear
unequal masses on the axis of symmetry. The remaining two masses are
symmetrically placed on both sides of the axis of symmetry. First four of the
masses effectively form a rhombus, with the fifth mass placed anywhere on the axis of symmetry.
For a second case, the five masses will form a triangle, with
one of the masses moved up on the axis of symmetry. As a particular case, we also discuss a four-body configuration with a zero central mass. The
paper is organized as follows: in section 2 we set up the problem and list
the main results, in sections 3, 4 and 5 we give analytical proofs and numerically
explore these results. In section 6, we give the Hamiltonian
equations of motion, regularize the singularities and then use
Poincar\'{e} surface of sections to investigate the phase space, analyzing the
effect of changing the mass of the central body on the stability of rhomboidal five-body problem. Conclusions are given in section 7.


\section{Main Results}

Let us consider $n-1$ mass points $(m_{0}, m_{1}, ..., m_{n-1})$, $m_i>0, i=\overline{0,n-1}$, the position vectors of $n-1$ mass points $r_i$, $i=\overline{0,n-1}$, and the inter-body distances $r_{ij}$, $i,j,=\overline{0,n-1}$.

Using \cite{Moeckel2014} notation, it is known that this $n$-body system forms a planar non-collinear central configuration if the following holds:
\begin{equation}
f_{ij}=\displaystyle\sum_{k=0,k\neq i,j}^{n-1}m_{k}(R_{ik}-R_{jk})\Delta _{ijk}=0,
\label{1.1}
\end{equation}%
where $R_{ij}=\frac{1}{r_{ij}^{3}}$ and $\Delta _{ijk}=(r_{i}-r_{j})\wedge
(r_{i}-r_{k})$. The $\Delta _{ijk}$ represent the areas of the
triangle determined by $(r_{i}-r_{j})$ and $(r_{i}-r_{k})$. For $n=5$ we obtain the general non-collinear five-body problem with the following ten CC equations:
\begin{eqnarray}
f_{01} &=&m_{2}(R_{02}-R_{12})\Delta _{012}+m_{3}(R_{03}-R_{13})\Delta
_{013}+  \nonumber \\
&+&m_{4}(R_{04}-R_{14})\Delta _{014},  \label{f01} \\
f_{02} &=&m_{1}(R_{01}-R_{21})\Delta _{021}+m_{3}(R_{03}-R_{23})\Delta
_{023}+  \nonumber \\
&+&m_{4}(R_{04}-R_{24})\Delta _{024},  \label{f02} \\
f_{03} &=&m_{1}(R_{01}-R_{31})\Delta _{031}+m_{2}(R_{02}-R_{32})\Delta
_{032}+  \nonumber \\
&+&m_{4}(R_{04}-R_{34})\Delta _{034},  \label{f03} \\
f_{04} &=&m_{1}(R_{01}-R_{41})\Delta _{041}+m_{2}(R_{02}-R_{42})\Delta
_{042}+  \nonumber \\
&+&m_{3}(R_{03}-R_{43})\Delta _{043},  \label{f04} \\
f_{12} &=&m_{0}(R_{10}-R_{20})\Delta _{120}+m_{3}(R_{13}-R_{23})\Delta
_{123}+  \nonumber \\
&+&m_{4}(R_{14}-R_{24})\Delta _{124},  \label{f12} \\
f_{13} &=&m_{0}(R_{10}-R_{30})\Delta _{130}+m_{2}(R_{12}-R_{32})\Delta
_{132}+  \nonumber \\
&+&m_{4}(R_{14}-R_{34})\Delta _{134},  \label{f13} \\
f_{14} &=&m_{0}(R_{10}-R_{40})\Delta _{140}+m_{2}(R_{12}-R_{42})\Delta
_{142}+  \nonumber \\
&+&m_{3}(R_{13}-R_{43})\Delta _{143},  \label{f14} \\
f_{23} &=&m_{0}(R_{20}-R_{30})\Delta _{230}+m_{1}(R_{21}-R_{31})\Delta
_{231}+  \nonumber \\
&+&m_{4}(R_{24}-R_{34})\Delta _{234},  \label{f23} \\
f_{24} &=&m_{0}(R_{20}-R_{40})\Delta _{240}+m_{1}(R_{21}-R_{41})\Delta
_{241}+  \nonumber \\
&+&m_{3}(R_{23}-R_{43})\Delta _{243},  \label{f24} 
\end{eqnarray}
\begin{eqnarray}
f_{34} &=&m_{0}(R_{30}-R_{40})\Delta _{340}+m_{1}(R_{31}-R_{41})\Delta
_{341}+  \nonumber \\
&+&m_{2}(R_{32}-R_{42})\Delta _{342}.  \label{f34}
\end{eqnarray}

\begin{lemma}
Dziobek equations \citep{Dziobek1900, Llibre2015, Libreetal2015} for a five-body problem when
the five masses have position vectors $\mathbf{r}_{0}=(0,w),$ $\mathbf{r}%
_{1}=(-1,0),$ $\mathbf{r}_{2}=(0,s),$ $\mathbf{r}_{3}=(1,0),$ $\mathbf{r}%
_{4}=(0,-t)$, where $s,t,w\in \mathbb{R}$ are
\begingroup\makeatletter
\check@mathfonts
\begin{eqnarray}
f_{01}&:=& m_{1}(R_{03}-R_{13})\Delta _{013}+m_{2}(R_{02}-R_{12})\Delta
_{012}+  \nonumber \\
&+&m_{4}(R_{04}-R_{14})\Delta _{014}=0,  \label{12a} \\
f_{12}&:=& m_{0}(R_{10}-R_{20})\Delta_{120}+m_{1}(R_{13}-R_{23})\Delta_{123}+
\nonumber \\
&+&m_{4}(R_{14}-R_{24})\Delta_{124}=0,  \label{12b} \\
f_{14}&:=& m_{0}(R_{10}-R_{40})\Delta _{140}+m_{1}(R_{13}-R_{43})\Delta
_{143}+  \nonumber \\
&+& m_{2}(R_{12}-R_{42})\Delta _{142}=0.  \label{12c}
\end{eqnarray}
\endgroup

\begin{proof}
Using the definition of $R_{ij}$, $\Delta _{ijk}$ and $\mathbf{r}_{i}$ $%
(i=0,1,2,3,4)$ we obtain
\begin{eqnarray}
R_{01} &=&R_{03}=\frac{1}{(1+w^{2})^{\frac{3}{2}}},\quad R_{02}=\frac{1}{%
|s-w|^{3}},  \label{R} \\
R_{13} &=&\frac{1}{8},R_{12}=R_{23}=\frac{1}{(1+s^{2})^{\frac{3}{2}}},R_{24}=%
\frac{1}{(s+t)^{3}},  \nonumber \\
R_{14} &=&R_{34}=\frac{1}{(1+t^{2})^{\frac{3}{2}}},\quad R_{04}=\frac{1}{%
|w+t|^{3}},  \nonumber
\end{eqnarray}%
and
\begin{eqnarray}\label{delta}
\Delta _{ijk} &=&-\Delta _{jik}=-\Delta _{ikj}=-\Delta _{kji},  \nonumber\\
\Delta _{ijk} &=&\Delta _{jki}=\Delta _{kij},  \nonumber \\
\Delta _{ijk} &=&0,\quad if\quad i=j\quad or\; i=k\; or\; j=k, \\
\Delta _{012} &=&\Delta _{023}=w-s,\quad \Delta _{013}=2w,  \nonumber \\
\Delta _{014} &=&\Delta _{043}=w+t,\quad \Delta _{132}=2s,\quad \Delta
_{143}=2t,  \nonumber \\
\Delta _{142} &=&\Delta _{324}=s+t,\quad \Delta _{024}=0.  \nonumber
\end{eqnarray}%
Using the symmetry of the problem through (\ref{R}-\ref{delta}), substituting $m_{1}=m_{3}$ in Eqs. \ref{f01}-\ref{f34}, and applying the corresponding relations
from equations (\ref{R}) and (\ref{delta}) we see that $f_{01}$
and $f_{03}$ are identical (Eqs. \ref{f01} and \ref{f03}). Similarly, $f_{12}$ is identical to $f_{23}$ (Eqs. \ref{f12} and \ref{f23}), and $f_{14}$
is identical to $f_{34}$ (Eqs. \ref{f14} and \ref{f34}). The
substitution of $R_{10}=R_{30}$, $R_{12}=R_{32}$, and $R_{14}=R_{34}$
implies that $f_{13}=0$. The substitution of $\Delta _{024}=\Delta
_{042}=\Delta _{240}=0$ in equations \ref{f02}, \ref{f04} and \ref{f24} gives %
\begingroup\makeatletter
\check@mathfonts
\begin{eqnarray}
f_{02} :=((R_{01}-R_{21})\Delta _{021}+(R_{03}-R_{23})\Delta _{023})m_{1}, &&
\label{13a} \\
f_{04} :=((R_{01}-R_{41})\Delta _{041}+(R_{03}-R_{43})\Delta _{043})m_{1}, &&
\label{13b} \\
f_{24} :=((R_{21}-R_{41})\Delta _{241}+(R_{23}-R_{43})\Delta _{243})m_{1}. &&
\label{13c}
\end{eqnarray}
\endgroup
Since $R_{01}=R_{03}$, $R_{23}=R_{21}$, $R_{41}=R_{43}$, and $\Delta
_{023}=-\Delta _{021}$, $\Delta _{041}=-\Delta _{043}$, $\Delta _{241}=-\Delta
_{243}$ we get $f_{02}=$ $f_{04}=$ $f_{24}=0$. We have shown that  $%
f_{01}=f_{03}$, $f_{12}=f_{23}$, $f_{14}=f_{34}$, $f_{02}=$ $%
f_{04}=f_{24}=f_{13}=0$. Consequently, $f_{01}$, $f_{12}$ and $f_{14}$
are the only necessary equations for general five-body problem (corresponding to \ref{12a}, \ref{12b}, \ref{12c}). This completes the proof of Lemma 1.
\end{proof}
\end{lemma}
\begin{figure}[!htb]
{\includegraphics[width=0.45\textwidth]{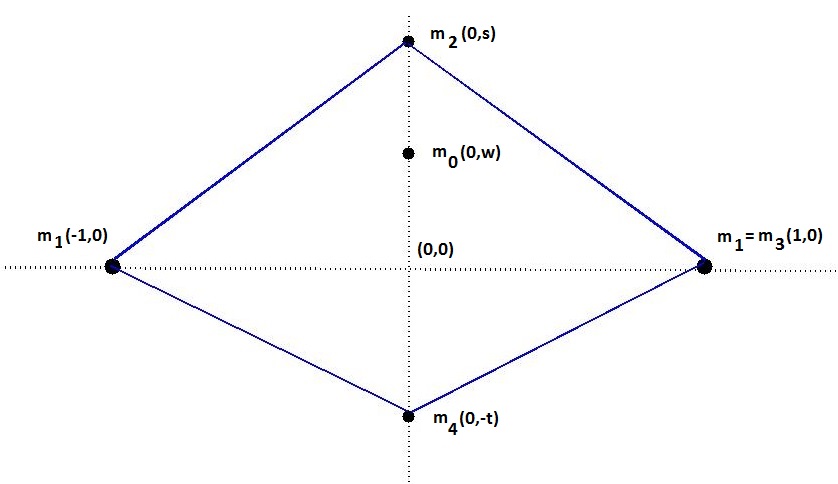}}
\caption{Rhomboidal five-body configurations.}
\label{S1-5}
\end{figure}
\begin{figure}[!htb]
{\includegraphics[width=0.45\textwidth]{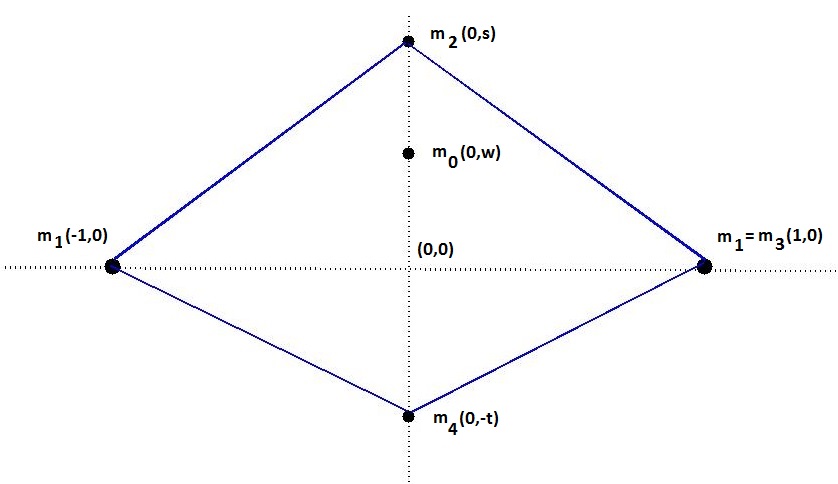}}
\caption{Triangular five-body configurations.}
\label{S2-5}
\end{figure}
\begin{theorem}
\label{theorems=t}Consider two pairs of equal masses on the vertices of a
rhombus with a fifth mass on one of the axis of symmetry. The five positive masses
have their position vectors $\mathbf{r}_{0}=(0,w),$ $\mathbf{r}_{1}=(-1,0),$ $%
\mathbf{r}_{2}=(0,t),$ $\mathbf{r}_{3}=(1,0),$ $\mathbf{r}_{4}=(0,-t)$,
where $t \in \mathbb{R}^+,w\in \mathbb{R}\backslash\{0\}$. Then their is no  central configuration of this type.
\end{theorem}

\begin{theorem}
\label{theorem3} Let $\mathbf{r}_{1}=(-1,0),$ $\mathbf{r}_{2}=(0,s),$ $%
\mathbf{r}_{3}=(1,0),$ $\mathbf{r}_{4}=(0,-t)$, where $s,t\in \mathbb{R}^{+}$. Then, the configuration $\mathbf{r}=(\mathbf{r}_{1},\mathbf{r}_{2},\mathbf{r%
}_{3},\mathbf{r}_{4})$ given in figure 1 will form a central configuration in the region
\begingroup\makeatletter
\check@mathfonts
\begin{eqnarray}
R_{\mu_2\mu_4}(s,t) &=&\{(s,t)|(0.268<t\leq 0.577 \wedge \nonumber\\
&\wedge& \frac{1-t^2}{2t} < s < \sqrt{3}) \vee (0.577 < t < \sqrt{3} \wedge \nonumber\\
&\wedge& \sqrt{t^{2}+1}-t < s < \sqrt{3})\}.
\label{rm2m4}
\end{eqnarray}
\endgroup
\end{theorem}


\begin{theorem}
\label{theorem4} Let $\mathbf{r}_{0}=(0,w),$ $\mathbf{r}_{1}=(-1,0),$ $%
\mathbf{r}_{2}=(0,s),$ $\mathbf{r}_{3}=(1,0),$ $\mathbf{r}_{4}=(0,-t)$,
where $s,t,w\in \mathbb{R}$.

\begin{enumerate}
\item[(a)] The configuration $\mathbf{r}=(\mathbf{r}_{0},\mathbf{r}_{1},%
\mathbf{r}_{2},\mathbf{r}_{3},\mathbf{r}_{4})$ will form a central
configuration with \begingroup\makeatletter
\check@mathfonts
\begin{eqnarray}
\mu _{0} &=&\frac{m_{0}}{m_{1}}=\frac{{A}_{{1}}{C}_{{2}}{C}_{{3}}-({B}_{{1}}{%
B}_{{3}}{C}_{{2}}+{B}_{{2}}{C}_{{1}}{C}_{{3}})}{{A}_{{2}}{C}_{{1}}{C}_{{3}}+{%
A}_{{3}}{B}_{{1}}{C}_{{2}}},  \nonumber
\end{eqnarray}
\begin{eqnarray}
\mu _{2} &=&\frac{m_{2}}{m_{1}}=\frac{{A}_{{3}}{B}_{{2}}{C}_{{1}}-({A}_{{1}}{%
A}_{{3}}{C}_{{2}}+{A}_{{2}}{B}_{{3}}{C}_{{1}})}{{A}_{{2}}{C}_{{1}}{C}_{{3}}+{%
A}_{{3}}{B}_{{1}}{C}_{{2}}}, {\label{11}}
\end{eqnarray}
\begin{eqnarray}
\mu _{4} &=&\frac{m_{4}}{m_{1}}=\frac{{A}_{{2}}{B}_{{1}}{B}_{{3}}-({A}_{{1}}{%
A}_{{2}}{C}_{{3}}+{A}_{{3}}{B}_{{1}}{B}_{{2}})}{{A}_{{2}}{C}_{{1}}{C}_{{3}}+{%
A}_{{3}}{B}_{{1}}{C}_{{2}}},  \nonumber
\end{eqnarray}
\endgroup
where
\begin{eqnarray*}
A_{1} &=&(R_{03}-R_{13})\Delta _{013}, \quad A_{2}=(R_{10}-R_{20})\Delta
_{120},  \nonumber \\
A_{3}&=&(R_{10}-R_{40})\Delta _{140}, \\
B_{1} &=&(R_{02}-R_{12})\Delta _{012}, \quad B_{2}=(R_{13}-R_{23})\Delta
_{123},  \nonumber \\
B_{3}&=&(R_{13}-R_{43})\Delta _{143}, \\
C_{1} &=&(R_{04}-R_{14})\Delta _{014}, \quad C_{2}=(R_{14}-R_{24})\Delta
_{124},  \nonumber \\
C_{3}&=&(R_{12}-R_{42})\Delta _{142}.
\end{eqnarray*}

\item[(b)] The mass ratios $\mu _{0}>0$, $\mu _{2}>0$ and $\mu _{4}>0$ form
a central configuration in
\begin{equation}
R(t,w)=R_{\mu _{0}}\cap R_{\mu _{2}}\cap R_{\mu _{4}},
\end{equation}%
where $R_{\mu _{0}}$, $R_{\mu _{2}}$ and $R_{\mu _{4}}$ are respectively
given by equations (\ref{Rm0}), (\ref{Rm2}), and (\ref{Rm4}).
\end{enumerate}
\end{theorem}


\section{Proof of Theorem \protect\ref{theorems=t}}

In order to prove Theorem 2, we write Dziobek's equations of Lemma 1 in a matrix form, after we substitute $m_{2}=m_{4}$:
\[
\left[
\begin{array}{ccc}
0 & B_{0} & C_{0} \\
A_{1} & B_{1} & C_{1} \\
A_{2} & -B_{1} & -C_{1}%
\end{array}%
\right] \cdot \left[
\begin{array}{c}
m_{0} \\
m_{1} \\
m_{2}%
\end{array}%
\right] =\left[
\begin{array}{c}
0 \\
0 \\
0%
\end{array}%
\right] ,
\]%
where
\begin{eqnarray*}
B_{0} &=&(R_{03}-R_{13})\Delta _{013},B_{1}=(R_{13}-R_{23})\Delta _{123}, \\
A_{1} &=&(R_{10}-R_{20})\Delta _{120},A_{2}=(R_{10}-R_{40})\Delta _{140}, \\
C_{1} &=&(R_{14}-R_{24})\Delta _{124}, \\
C_{0} &=&(R_{02}-R_{12})\Delta _{012}+(R_{04}-R_{14})\Delta _{014}.
\end{eqnarray*}%
For the above system of equation, in order to have a non-trivial solution, the determinant of the
augmented matrix must be zero. After adding row 2 with row 3 of the
augmented matrix, the third row reduces to $\{A_{1}+A_{2},0,0\}$.
Hence, $f_{A_{1}A_{2}}=A_{1}+A_{2}=0$ will imply the existence of non-trivial
solution (the case $m_{0}=0$ will be discussed in detail in section 4). Therefore, from the remaining two equations
\begin{equation}
\mu _{1}=\frac{m_{1}}{m_{2}}=-\frac{C_{0}}{B_{0}},\text{ and \ }\mu _{0}=%
\frac{m_{0}}{m_{2}}=\frac{C_{1}B_{0}-C_{0}B_{1}}{B_{0}A_{2}}
\end{equation}%
we obtain that they will form a central configuration with respect to the geometric constraint
\begin{eqnarray}
f_{A_{1}A_{2}} &=&(R_{01}-R_{02})\Delta _{012}+(R_{01}-R_{04})\Delta _{014}
\nonumber \\
&=&\frac{t-w}{|t-w|^{3}}-\frac{t+w}{|t+w|^{3}}+\frac{2w}{(w^{2}+1)^{\frac{3}{%
2}}}=0.
\end{eqnarray}%
For positive solutions we must have $f_{A_{1}A_{2}}=0$ and $\mu _{1}>0$, $\mu _{2}>0$.

\begin{enumerate}
\item[a.] If $t>|w|, w\in \mathbb{R}\backslash \{0\}$ then $f_{A_{1}A_{2}}=0$
can be written as a polynomial in two variables $t$ and $w$ such that $t\neq \pm w$.
\begin{eqnarray}
G_{1}(w) & := & 2t^{4}w-4t^{2}w^{3}+4tw(w^{2}+1)^{3/2}+  \nonumber \\
&+&2w^{5}=0.
\end{eqnarray}
The polynomial $G_{1}(w)$ has two real negative roots and two complex roots.
As we require $t>0$, therefore this case is not interesting.
\begin{figure}[!htb]
\includegraphics[width=0.44\textwidth]{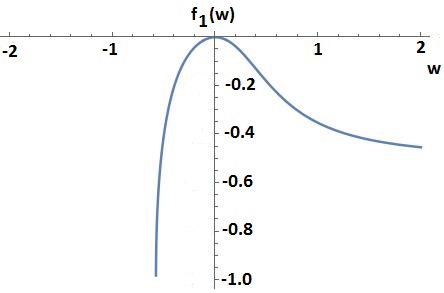} %
\includegraphics[width=0.44\textwidth]{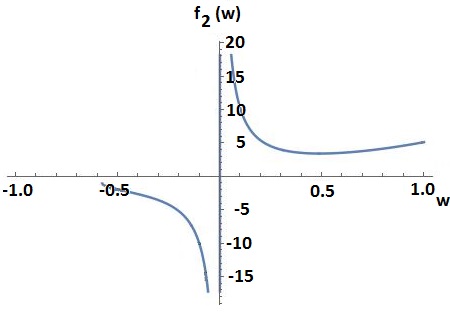}
\caption{a) $f_1(w)$ \hskip3cm b) $f_2(w)$ }
\label{f1w}
\end{figure}
\item[b.] If $t<w,w>0$ \ then $f_{A_{1}A_{2}}=0$ can be written as a
polynomial in two variables $t$ and $w$ such that $t\neq \pm w\neq 0$.
\begin{eqnarray}
G_{2}(w)&:=&t^{4}w-t^{2}(2w^{3}+(1+{w^{2}})^{3/2})+  \nonumber \\
&+&w^{2}(w^{3}-(1+{w^{2}})^{3/2})=0.
\end{eqnarray}
The polynomial $G_{2}(w)$ is quadratic in $t^{2}$ and will have four roots
as functions of $w$.
\begingroup\makeatletter\def\f@size{8}
\check@mathfonts
\[
t(w)=\pm \sqrt{\frac{(w^{2}+1)^{3/2}}{2w}+w^{2}\pm \frac{\sqrt{%
(w^{2}+1)^{3}+8w^{3}(w^{2}+1)^{3/2}}}{2w}}.
\]
\endgroup

The function \begingroup\makeatletter\def\f@size{8}
\check@mathfonts
\[
f_{1}(w)=\frac{(w^{2}+1)^{3/2}}{2w}+w^{2}-\frac{\sqrt{%
(w^{2}+1)^{3}+8w^{3}(w^{2}+1)^{3/2}}}{2w}
\]%
\endgroup
is negative for all values of $w$, see figure \ref{f1w}a. The function %
\begingroup\makeatletter\def\f@size{8}
\check@mathfonts
\[
f_{2}(w)=\frac{(w^{2}+1)^{3/2}}{2w}+w^{2}+\frac{\sqrt{%
(w^{2}+1)^{3}+8w^{3}(w^{2}+1)^{3/2}}}{2w}
\]%
\endgroup
is positive for $w>0$ (see figure \ref{f1w}b). Then $G_{2}(w)$ has two real
roots with one positive and one negative root. The negative root is not interesting.

 The necessary condition for the existence of central
configurations is satisfied at
\[
t_{1}(w)=\sqrt{f_{2}(w)}.
\]%
\begin{figure}[!htb]
\includegraphics[width=0.44\textwidth]{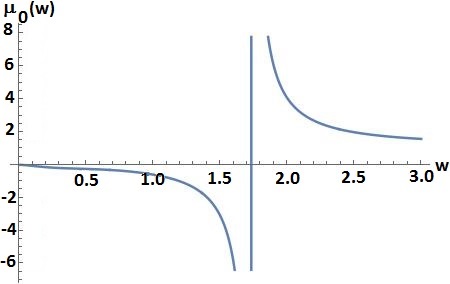} %
\includegraphics[width=0.44\textwidth]{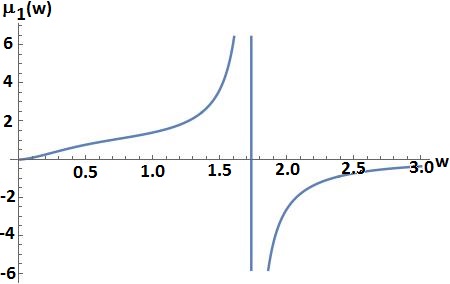}
\caption{a) $\protect\mu_0(w)$ when $t=t_1(w)$ \hskip 0.5cm b) $\protect\mu%
_1(w)$ when $t=t_1(w)$ }
\label{mu0mu1Att=t1}
\end{figure}
It is straightforward to see from figure \ref{mu0mu1Att=t1} that $\mu
_{0}(w)>0$ when $w<\sqrt{3}$, and $\mu _{1}(w)>0$ when $w>\sqrt{3}$. So,
there is no common region where both $\mu _{0}(w)$ and $\mu _{1}(w)$ are
positive (Fig. 5).
\begin{figure}[!htb]
{\includegraphics[width=0.44\textwidth]{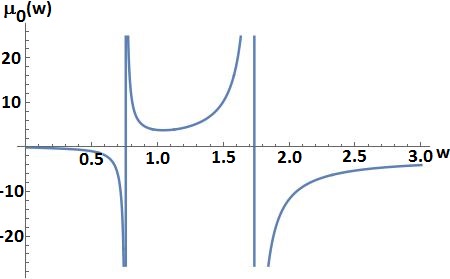} %
\includegraphics[width=0.44\textwidth]{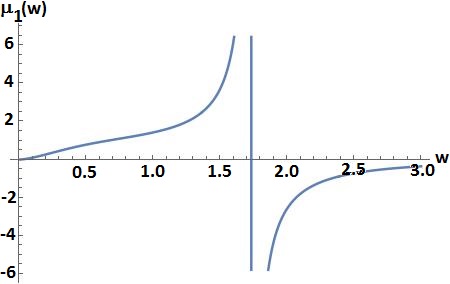}}
\caption{a) $\protect\mu _{0}(w)$ when $t=-t_{1}(w)$ \hskip0.5cm b) $\protect%
\mu _{1}(w)$ when $t=-t_{1}(w)$ }
\label{mu0mu1Att=-t1}
\end{figure}

\item[c.] If $t<|w|, w<0$ then $f_{A_{1}A_{2}}=0$ can be written as a
polynomial in two variables $t$ and $w$ such that $t\neq \pm w\neq 0$. \begingroup%
\makeatletter
\begin{eqnarray}
G_{3}(w)&:=&t^{4}w+t^{2}( ( w^{2}+1) ^{3/2}-2w^{3}) +\nonumber \\
&+& w^{2}(w^{3}+( w^{2}+1)^{3/2})=0.
\end{eqnarray}
\endgroup
\end{enumerate}

The polynomial $G_{3}(w)$ is quadratic in $t^{2}$ and have the following
four real roots as functions of $w$.
\begin{eqnarray}
t_{2}(w)=-\sqrt{f_{3}(w)},\quad t_{3}(w)=\sqrt{f_{3}(w)},  \nonumber \\
t_{4}(w)=-\sqrt{f_{4}(w)},\quad t_{5}(w)=\sqrt{f_{4}(w)},  \nonumber
\end{eqnarray}
where
\begingroup\makeatletter\def\f@size{8}
\check@mathfonts
\[
f_{3}(w)=-\frac{(w^{2}+1)^{3/2}}{2w}+w^{2}-\frac{\sqrt{%
(w^{2}+1)^{3/2}\cdot(1-8w^{3})}}{2w},
\]
\[
f_{4}(w)=-\frac{(w^{2}+1)^{3/2}}{2w}+w^{2}+\frac{\sqrt{%
(w^{2}+1)^{3/2}\cdot(1-8w^{3})}}{2w}.
\]
\endgroup
\begin{figure}[!htb]
\includegraphics[width=0.45\textwidth]{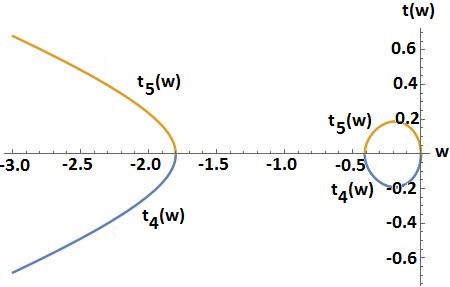} %
\includegraphics[width=0.45\textwidth]{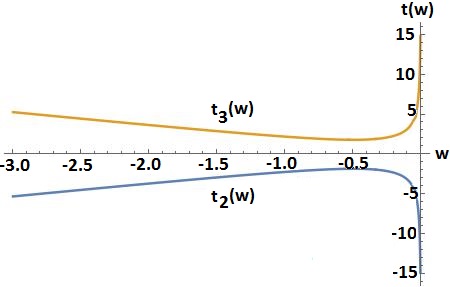}
\caption{Roots of $G_{3}(w)$}
\label{rootsG3}
\end{figure}
Two of the above roots i.e. $t_{3}(w)$ and $t_{5}(w)$ are positive, and $t_{2}(w)$ and $t_{4}(w)$ are negative. These roots are given in figure \ref{rootsG3}.

By the close inspection of figure \ref{rootsG3} it can be seen that $%
t_{3}(w)>|w|$ for all $w<0$. As a result, this root is invalid as we have
the additional constraint of $t<|w|$ for this special case. The
positive root $t_{5}(w)$ remains less than $|w|$ for $w<0$ when it is
defined in $(-4,-1.8)$ and $(-0.4,0)$.

The mass ratio $\mu_1$ is positive when $w\in (-0.2,0)$ as $B_0$ and $C_0$
have opposite signs in this interval, see figure (\ref{muit=t5}a). In the
same interval $A_2B_0$ and $B_0C_1-B_1C_0$ have opposite signs and hence $%
\mu_0$ is negative, no central configurations at $t=t_5(w)$ when $w\in
(-0.4, 0)$. For $w<-1.8$, $C_0$ is an increasing function of $w$ with its
absolute minimum, 0, occurring when $w \to -\infty$. Hence, $C_0>0$ for all $%
w<-1.8$. Similarly, $B_0$ is also positive. Therefore, $\mu_1<0$ for $w<-1.8$
(see figure \ref{muit=t5}b). This completes the proof of Theorem \ref%
{theorems=t}.

\begin{figure}[!htb]
\includegraphics[width=0.45\textwidth]{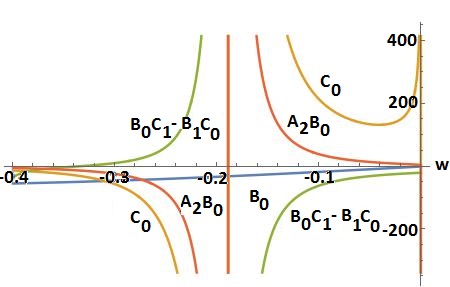} %
\includegraphics[width=0.45\textwidth]{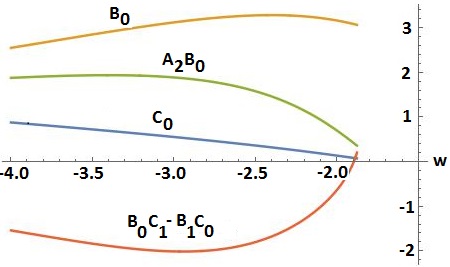}
\caption{The curves $B_0$, $C_0$, $A_2B_0$ and $B_0C_1-B_1C_0$ at $%
t=t_{5}(w) $ such that a) $w\in(-0.4,0)$ b) $w\in(-4,-1.8)$. }
\label{muit=t5}
\end{figure}

\section{Proof of Theorem \protect\ref{theorem3}}

After substituting $m_{0}=0$ in the equations of Lemma 1 and using the
resulting symmetries, the following reduced system of Dziobek's equations is
obtained:
\begin{eqnarray}
(R_{13}-R_{23})\Delta _{123}m_{1}+(R_{14}-R_{24})\Delta _{124}m_{4} &=&0,
\nonumber \\
(R_{13}-R_{14})\Delta _{134}m_{1}+(R_{12}-R_{24})\Delta _{124}m_{2} &=&0,
\label{eqtn4bp}
\end{eqnarray}%
where
\begin{eqnarray*}
R_{23} &=&(1+s^{2})^{-3/2},\quad R_{14}=(1+t^{2})^{-3/2}, \\
R_{24} &=&(t+s)^{-3}, \\
\Delta _{123} &=&-2s,\quad \Delta _{124}=-(s+t),\quad \Delta _{134}=-2t.
\end{eqnarray*}%
As the above system is under determined, therefore we define $\mu
_{2}=m_{2}/m_{1}$ and $\mu _{4}=m_{4}/m_{1}$. Simultaneous solution of
equations (\ref{eqtn4bp}) gives
\begin{eqnarray}
\mu _{2} &=&\frac{(s^{2}+1)^{3/2}t(s+t)^{2}f_{5}(t)}{%
4(t^{2}+1)^{3/2}f_{6}(s,t)}, \\
\mu _{4} &=&\frac{sf_{7}(s)(t^{2}+1)^{3/2}(s+t)^{2}}{%
4(s^{2}+1)^{3/2}f_{8}(s,t)},
\end{eqnarray}%
where
\begin{eqnarray}
f_{5}(t)&=&((t^{2}+1)^{3/2}-8),\nonumber\\
f_{6}(s,t)&=&((s^{2}+1)^{3/2}-(s+t)^{3}),\nonumber\\
f_{7}(s)&=&((s^{2}+1)^{3/2}-8),\nonumber\\
f_{8}(s,t)&=&((t^{2}+1)^{3/2}-(s+t)^{3}).\nonumber
\end{eqnarray}
To find the central configuration (CC) regions, where all the masses are
positive, we need to find regions in the $st$-plane, where both $\mu _{2}$
and $\mu _{4}$ are positive. As $f_{5}(t)$ is a monotonically increasing
function of $t$ for all $t>0$, with a single zero at $t=\sqrt{3}$, therefore
it is straightforward to see that $f_{5}(t)<0$, when $t<\sqrt{3}$. Solving
$f_{6}(s,t)=0$ for variable $t$, we get $t(s)=-s+\sqrt{s^{2}+1}$. It is easy to check
that $f_{6}(s,t)<0$ when $t(s)>-s+\sqrt{s^{2}+1}$. The two functions
$f_{5}(t)$ and $f_{6}(s,t)$ are never simultaneously positive, therefore
$\mu_{2}>0$ in
\begin{equation}
R_{\mu _{2}}=\{(s,t)|t>0\wedge s>0\wedge \sqrt{s^{2}+1}-s<t<\sqrt{3}\}.
\end{equation}%
Similarly, $f_{7}(s)=0$, when $s=\sqrt{3}$ and $f_{8}(s,t)=0$, when $s(t)=-t+%
\sqrt{t^{2}+1}$. Similar $\mu _{4}>0$ in
\begin{equation}
R_{\mu _{4}}=\{(s,t)|s>0\wedge t>0\wedge \sqrt{t^{2}+1}-t<s<\sqrt{3}\}.
\end{equation}
The intersection of $R_{\mu _{2}}$ and $R_{\mu _{4}}$ is $R_{\mu _{2}\mu
_{4}}$ given by the equation (\ref{rm2m4}) (see Fig. \ref{CC4BPRhombus}).

\begin{corollary}
Consider $t<0$ in the setup of Theorem \ref{theorem3}, guaranteeing a
triangular four-body arrangement, then the configuration
$(m_{1}, m_{2}, m_{1}, m_{4})$ will form a triangular central configuration in %
\begingroup\makeatletter
\check@mathfonts
\begin{eqnarray*}
TR_{t_{-}}(s,t) &=&\{(s,t)|(-3.73<t<-\sqrt{3}\wedge \\
&\wedge& h_{1}(t)<s<\sqrt{3}) \vee (-\sqrt{3} <t\leq -1 \wedge \\
&\wedge& (0<s<h_{1}(t) \vee \sqrt{3}<s<h_{2}(t))) \vee \\
&\vee& (-1 <t<-\frac{1}{\sqrt{3}}\wedge \sqrt{3}<s<h_{2}(t)) \\
&\vee& (-\frac{1}{\sqrt{3}} <t<0\wedge h_{2}(t)<s<\sqrt{3})\},
\end{eqnarray*}%
\endgroup
where
\[
h_{1}(t)=\frac{1-t^{2}}{2t},\quad h_{2}(t)=\sqrt{t^{2}+1}-t.
\]
\end{corollary}

\begin{proof}
The proof of the Corollary 5 follows the same procedure as the Theorem 3.
Therefore, we only give a sketch of this proof and leave the details
to the interested readers.

Using the same procedure as above we can show that the mass $m_{2}$ is
positive in the region%
\begingroup\makeatletter
\check@mathfonts
\begin{eqnarray*}
TR_{m_2}(s,t) &=&\{(s,t)|(0<s<\frac{1}{\sqrt{3}}\wedge(-\sqrt{3}<t<H(s)
\\
&\vee& -s<t<0)) \vee (\frac{1}{\sqrt{3}}<s<\sqrt{3} \wedge \\
&\wedge& (H(s)<t<-\sqrt{3}\vee -s<t<0)) \vee \\
&\vee& (s>\sqrt{3}\wedge (H(s)<t<-s\vee\\
&\vee& -\sqrt{3}<t<0))\},
\end{eqnarray*}
\endgroup
where
\[
H(s)=-\sqrt{s^{2}+1}-s.
\]

Similarly, $m_{4}$ is positive in the region
\begingroup\makeatletter\def\f@size{8}
\check@mathfonts
\begin{eqnarray*}
TR_{m_{4}}(s,t) &=&(t<-\sqrt{3}\wedge (0<s<\sqrt{3}\vee -t<s<h_{2}(t))) \\
&\vee& (-\sqrt{3}<t<-\frac{1}{\sqrt{3}}\wedge (0<s<-t\vee \\
&\vee& \sqrt{3}<s<h_{2}(t))) \vee (-\frac{1}{\sqrt{3}}<t<0\wedge \\
&\wedge& (0<s<-t\vee h_{2}(t)<s<\sqrt{3})).
\end{eqnarray*}
\endgroup
The intersection of $TR_{m_{2}}(s,t)$ and $TR_{m_{4}}(s,t)$ give $%
TR_{t_{-}}(s,t).$
\end{proof}

\begin{figure}[!htb]
\includegraphics[width=0.45\textwidth]{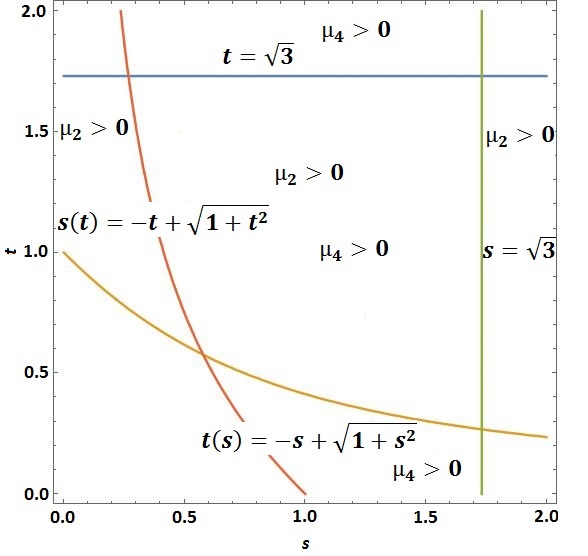}
\caption{Central Configuration region in the case of Rhomboidal four-body
problem }
\label{CC4BPRhombus}
\end{figure}


\section{Proof of Theorem \protect\ref{theorem4}}


\subsection{Proof of Theorem \protect\ref{theorem4}a}

Equations (\ref{12a}), (\ref{12b}) and (\ref{12c})
of Lemma 1 define the CC of rhomboidal or triangular five-body problems.
Let $\mu _{0}=\frac{m_{0}}{m_{1}}$, $\mu _{2}=\frac{m_{2}}{m_{1}}$, $%
\mu _{4}=\frac{m_{4}}{m_{1}}$ and rewrite the equations from Lemma 1 as
\begin{eqnarray*}
A_{1}m_{1}+B_{1}m_{2}+C_{1}m_{4} &=&0, \\
A_{2}m_{0}+B_{2}m_{1}+C_{2}m_{4} &=&0, \\
A_{3}m_{0}+B_{3}m_{1}+C_{3}m_{2} &=&0,
\end{eqnarray*}%
where
\begin{eqnarray*}
A_{1} &=&(R_{03}-R_{13})\Delta _{013},A_{2}=(R_{10}-R_{20})\Delta _{120}, \\
A_{3} &=&(R_{10}-R_{40})\Delta _{140}, \\
B_{1} &=&(R_{02}-R_{12})\Delta _{012},B_{2}=(R_{13}-R_{23})\Delta _{123}, \\
B_{3} &=&(R_{13}-R_{43})\Delta _{143}, \\
C_{1} &=&(R_{04}-R_{14})\Delta _{014},C_{2}=(R_{14}-R_{24})\Delta _{124}, \\
C_{3} &=&(R_{12}-R_{42})\Delta _{142}.
\end{eqnarray*}%
Divide the above equations by $m_{1}$ and write them in the following
matrix form:
\begin{equation}
\left[
\begin{array}{ccc}
0 & B_{1} & C_{1} \\
A_{2} & 0 & C_{2} \\
A_{3} & C_{3} & 0%
\end{array}%
\right] \cdot \left[
\begin{array}{c}
\mu _{0} \\
\mu _{2} \\
\mu _{4}%
\end{array}%
\right] =-\left[
\begin{array}{c}
A_{1} \\
B_{2} \\
B_{3}%
\end{array}%
\right].  \label{14}
\end{equation}%
After a number of row operation, a simultaneous solution of the above system
is obtained as below.%
\begin{eqnarray}
\mu _{0} &=&\frac{{A}_{{1}}{C}_{{2}}{C}_{{3}}-({B}_{{1}}{B}_{{3}}{C}_{{2}}+{B%
}_{{2}}{C}_{{1}}{C}_{{3}})}{{A}_{{2}}{C}_{{1}}{C}_{{3}}+{A}_{{3}}{B}_{{1}}{C}%
_{{2}}},  \nonumber \\
\mu _{2} &=&\frac{{A}_{{3}}{B}_{{2}}{C}_{{1}}-({A}_{{1}}{A}_{{3}}{C}_{{2}}+{A%
}_{{2}}{B}_{{3}}{C}_{{1}})}{{A}_{{2}}{C}_{{1}}{C}_{{3}}+{A}_{{3}}{B}_{{1}}{C}%
_{{2}}}, \\
\mu _{4} &=&\frac{{A}_{{2}}{B}_{{1}}{B}_{{3}}-({A}_{{1}}{A}_{{2}}{C}_{{3}}+{A%
}_{{3}}{B}_{{1}}{B}_{{2}})}{{A}_{{2}}{C}_{{1}}{C}_{{3}}+{A}_{{3}}{B}_{{1}}{C}%
_{{2}}}.  \nonumber
\end{eqnarray}%
This completes the proof of Theorem 4a.

\subsection{Proof of Theorem \protect\ref{theorem4}b}

The required values of $R_{ij}$ and $\Delta _{ijk}$ can be generated from
Lemma 1 by the substitution of $s=1$.
\begin{lemma}
\label{lemma5} The function
\[
\mu _{0}=\frac{{A}_{{1}}{C}_{{2}}{C}_{{3}}-({B}_{{1}}{B}_{{3}}{C}_{{2}}+{B}_{%
{2}}{C}_{{1}}{C}_{{3}})}{{A}_{{2}}{C}_{{1}}{C}_{{3}}+{A}_{{3}}{B}_{{1}}{C}_{{%
2}}}
\]%
attains positive value in the region
\begin{equation}
R_{\mu _{0}}=(R_{D}\cap R_{N_{\mu _{0}}})\cup (R_{D}^{c}\cap R_{N_{\mu
_{0}}}^{c}),  \label{Rm0}
\end{equation}%
where $R_{D}^{c}$ and $R_{N_{\mu _{0}}}^{c}$ are the complements of the
regions $R_{D}$ and $R_{N_{\mu _{0}}}$.
\end{lemma}

\begin{proof}
Let $N_{\mu _{0}}={A}_{{1}}{C}_{{2}}{C}_{{3}}-({B}_{{1}}{B}_{{3}}{C}_{{2}}+{B%
}_{{2}}{C}_{{1}}{C}_{{3}})$ and $D={A}_{{2}}{C}_{{1}}{C}_{{3}}+{A}_{{3}}{B}_{%
{1}}{C}_{{2}}.$ For $\mu _{0}$ to be positive, $N_{\mu _{0}}$ and $D$ must
have the same sign.

The denominator $D$ will be positive when both factors ${A}_{{2}}{C}_{{1}}{C}%
_{{3}}$ and $A_{3}B_{1}C_{2}$ are positive or at least one of them is
positive, such that the positive part is greater than the absolute value of
the negative part. Both these factors are positive in
\begingroup
\makeatletter
\begin{eqnarray}
R_{da}(t,w) &=&\{(t,w)|(0<w\leq 0.41\wedge d_{1}(w)<t<1)  \nonumber \\
&\vee &(0.41<w<0.58\wedge d_{1}<t<0.5 \cdot d_{2}(w))  \nonumber \\
&\vee &(w>2.41\wedge 0.41<t<1)\},
\end{eqnarray}%
\endgroup where
\[
d_{1}(w)=\sqrt{w^{2}+1}-w,\quad d_{2}(w)=\frac{(1-w^{2})}{w}.
\]%
Similarly, when

\begin{enumerate}
\item ${A}_{{2}}{C}_{{1}}{C}_{{3}}>0$ and $A_{3}B_{1}C_{2}<0$ then ${A}_{{2}}{C}_{{1}}{C}_{{3}}$ $>$ $|A_{3}B_{1}C_{2}|$ in the following region. \begingroup%
\makeatletter
\begin{eqnarray*}
R_{db}(t,w) &=&\{(t,w)|(0.4<t\leq 0.58 \wedge \\
&\wedge& 0.6\cdot d_{2}(t) <w<1)\vee \\
&\vee& (0.58<t<1 \wedge d_{1}(t)<w<1)\}.
\end{eqnarray*}%
\endgroup

\item ${A}_{{2}}{C}_{{1}}{C}_{{3}}<0$ and $A_{3}B_{1}C_{2}>0$ then
$A_{3}B_{1}C_{2}$ $>$ $|{A}_{{2}}{C}_{{1}}{C}_{{3}}|$ in the following region. %
\begingroup\makeatletter
\begin{eqnarray*}
R_{dc}(t,w) &=&\{(t,w)|(w\leq -1.49\wedge 0<t<0.23) \\
&\vee &(-1.49<w<-1\wedge\\
&\wedge & 0<t<0.2 \cdot d_{2}(w)) \vee \\
&\vee &(1<w<2.41\wedge 0.41<t\leq 1)\}.
\end{eqnarray*}%
\endgroup
\end{enumerate}
Therefore, the denominator $D$ is positive in  (see Fig. \ref{Regions-m0}a):
\begin{equation}
R_{D}=R_{da}(t,w)\cup R_{db}(t,w)\cup R_{dc}(t,w).  \label{RD}
\end{equation}%

The numerator of $\mu _{0}$, $N_{\mu _{0}}$ is positive in  (see Fig. \ref{Regions-m0}b and Appendix):
\begin{equation}
R_{N_{\mu _{0}}}=R_{cN_{\mu _{0}}}(t,w)\cup R_{dN_{\mu _{0}}}(t,w).
\end{equation}%
Therefore, $\mu _{0}>0$ in the intersection of $R_{D}$
and $R_{N_{\mu _{0}}}$, and the
intersection of their complements. Region $R_{\mu _{0}}$ is shown in figure
\ref{Regions-m0}c. This completes the proof of Lemma \ref{lemma5}.
\begin{figure}[!htb]
\centering
\includegraphics[width=0.23\textwidth]{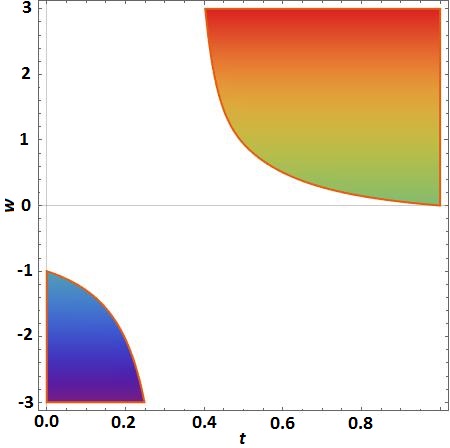}
\includegraphics[width=0.23\textwidth]{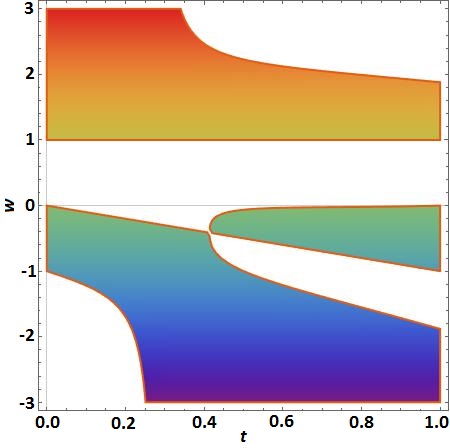}
\includegraphics[width=0.45\textwidth]{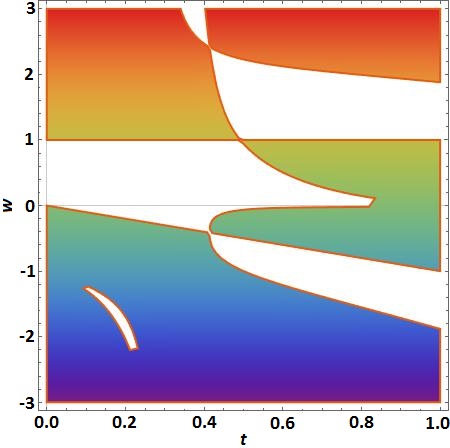}
\caption{Regions a)$R_{D}$ (colored) b) $R_{N_{\protect\mu _{0}}}$ (colored)
c) $R_{\protect\mu _{0}}$ (colored).}
\label{Regions-m0}
\end{figure}

\end{proof}

\begin{lemma}
\label{lemma6} The function
\[
\mu _{2}=\frac{A_{3}B_{2}C_{1}-(A_{1}A_{3}C_{2}+A_{2}B_{3}C_{1})}{{A}_{{2}}{C%
}_{{1}}{C}_{{3}}+{A}_{{3}}{B}_{{1}}{C}_{{2}}}
\]%
attains positive values in the region
\begin{equation}
R_{\mu _{2}}=(R_{D}\cap R_{N_{\mu _{2}}})\cup (R_{D}^{c}\cap R_{N_{\mu
_{2}}}^{c}),  \label{Rm2}
\end{equation}%
where \ $R_{D}^{c}$ and $R_{N_{\mu _{2}}}^{c}$ are complements of the
regions $R_{D}$ and $R_{N_{\mu _{2}}}$ respectively.
\end{lemma}

\begin{proof}
Let $N_{\mu _{2}}=A_{3}B_{2}C_{1}-(A_{1}A_{3}C_{2}+A_{2}B_{3}C_{1}).$ For $%
\mu _{2}$ to be positive, $N_{\mu _{2}}$ and $D$ must have the same sign. We
have given the complete analysis of $D$ in Lemma \ref{lemma5}. Now, we need
to find regions where $N_{\mu _{2}}$ is positive. It is proved in the appendix that   $N_{\mu _{2}}(t,w)>0$ in
\begin{equation}
R_{N_{\mu _{2}}}=R_{aN_{\mu _{2}}}(t,w)\cup R_{bN_{\mu _{2}}}(t,w)\cup
R_{cN_{\mu _{2}}}(t,w).  \label{RNmu2}
\end{equation}%
\newline
Region $R_{N_{\mu _{2}}}$ is shown in figure \ref{Regions-m2}a.
Consequently, $\mu _{2}>0$ in the intersection of $R_{D}$ and $R_{N_{\mu
_{2}}}$, as well as the intersection of their complements. Region $R_{\mu
_{2}}$ is shown in figure \ref{Regions-m2}b. This completes the proof of
Lemma \ref{lemma6}.
\begin{figure}[!htb]
\centering
\includegraphics[width=0.23\textwidth]{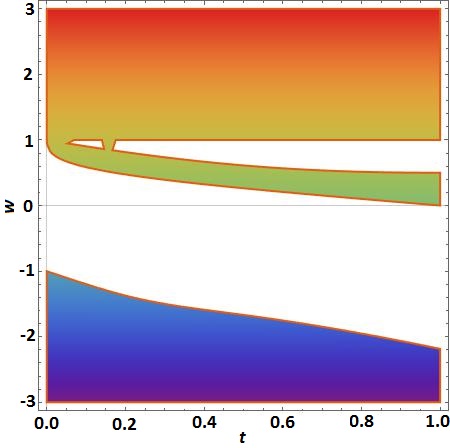} %
\includegraphics[width=0.23\textwidth]{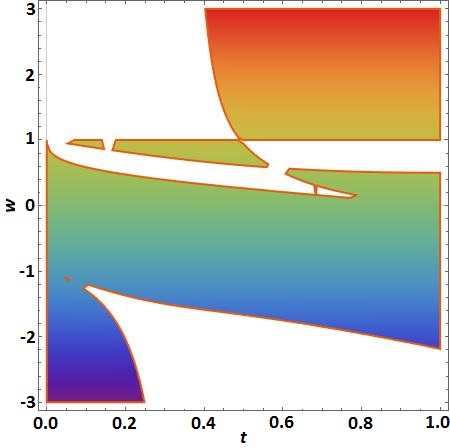}
\caption{Regions a) $R_{N_{\protect\mu _{2}}}$ (colored) b) $R_{\protect\mu %
_{2}}$ (colored). }
\label{Regions-m2}
\end{figure}

\end{proof}


\begin{lemma}
\label{lemma7} The function
\[
\mu_{4}=\frac{A_{2}B_{1}B_{3}-(A_{3}B_{1}B_{2}+A_{1}A_{2}C_{3})} {%
A_{2}C_{3}C_{1}+A_{3}B_{1}C_{2}}
\]
attains positive values in the region
\begin{equation}
R_{\mu _{4}}=(R_{D}\cap R_{N_{\mu _{4}}})\cup (R_{D}^{c}\cap R_{N_{\mu
_{4}}}^{c}).  \label{Rm4}
\end{equation}
\end{lemma}

\begin{proof}
Let $N_{\mu _{4}}=A_{2}B_{1}C_{3}-(A_{3}B_{1}B_{2}+A_{1}A_{2}C_{3}).$ For
positive $\mu _{4}$, $N_{\mu 4}$ and $D$ must have the same sign. We
have given the complete analysis of $D$ in Lemma \ref{lemma5}. Now, we need
to find regions where $N_{\mu _{4}}$ is positive. It is shown in the appendix that $N_{\mu _{4}}(t,w)$ is positive in  (see Fig. \ref{Regions-m4}a):
\begin{equation}
R_{N_{\mu _{4}}}=R_{aN_{\mu _{4}}}(t,w)\cup R_{bN_{\mu _{4}}}(t,w)\cup
R_{cN_{\mu _{4}}}(t,w).  \label{RNmu4}
\end{equation}%
Thus, $\mu _{4}>0$ in the intersection of $R_{D}$ and $R_{N_{\mu _{4}}}$,
as well
as the intersection of their complements. Region $R_{\mu _{4}}$ is shown in
figure \ref{Regions-m4}b. This completes the proof of Lemma \ref{lemma7}.
\end{proof}
\begin{figure}[!tbh]
\centering
\includegraphics[width=0.23\textwidth]{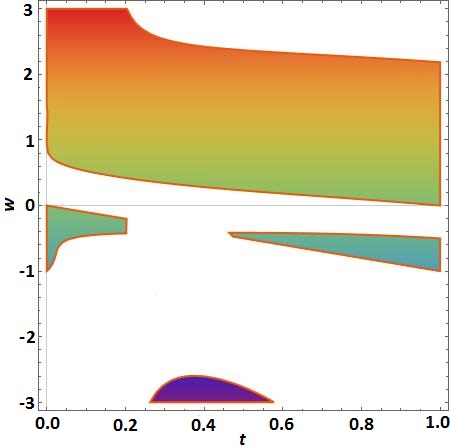} %
\includegraphics[width=0.23\textwidth]{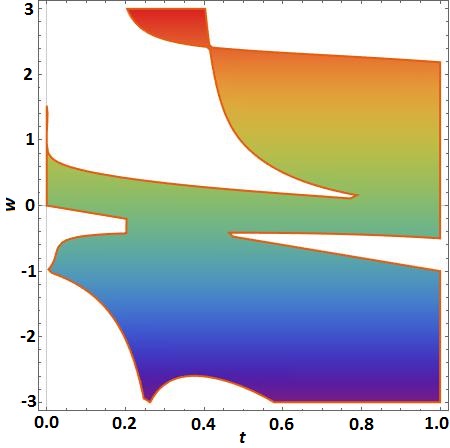}
\caption{Regions a) $R_{N_{\protect\mu _{4}}}$ (colored) b) $R_{\protect\mu %
_{4}}$ (colored).}
\label{Regions-m4}
\end{figure}

The central configuration region where all the masses have positive values
is determined by
\begin{equation}
R(t,w)=R_{\mu _{0}}\cap R_{\mu _{2}}\cap R_{\mu _{4}}.
\end{equation}%
The region $R(t,w)$ is given in figure \ref{Regions-cc}. This completes the
proof of Theorem \ref{theorem4}b. To better understand the nature of the complicated CC regions in  figure \ref{Regions-cc}, a number of relevant examples are given in figures  \ref{Rhombus1} and  \ref{Rhombus2}.
\begin{figure}[!htb]
\centering
{\ \includegraphics[width=0.45\textwidth]{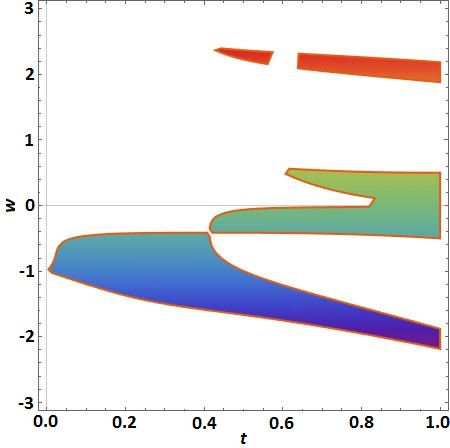}}
\caption{Region ($R(t,w)$) of central configuration in the case of Rhomboidal 5-body
problem.}
\label{Regions-cc}
\end{figure}

\begin{figure}[!htb]
\centering
\includegraphics[width=0.45\textwidth]{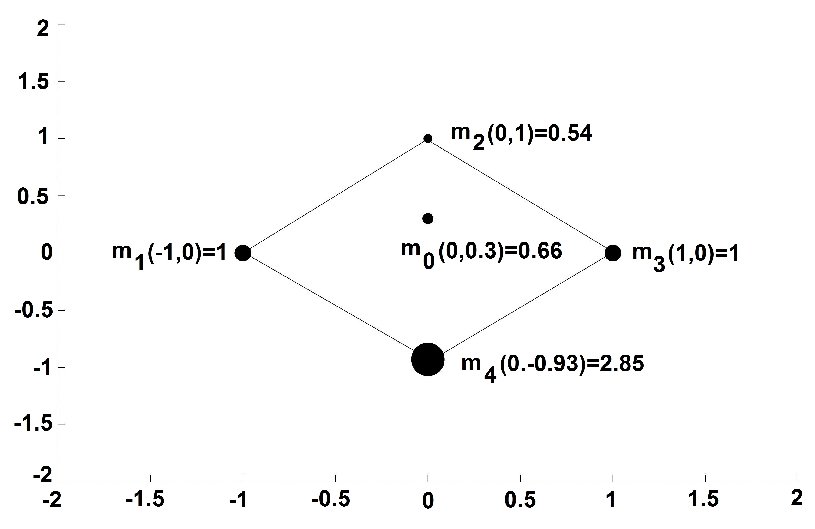}\\
\includegraphics[width=0.45\textwidth]{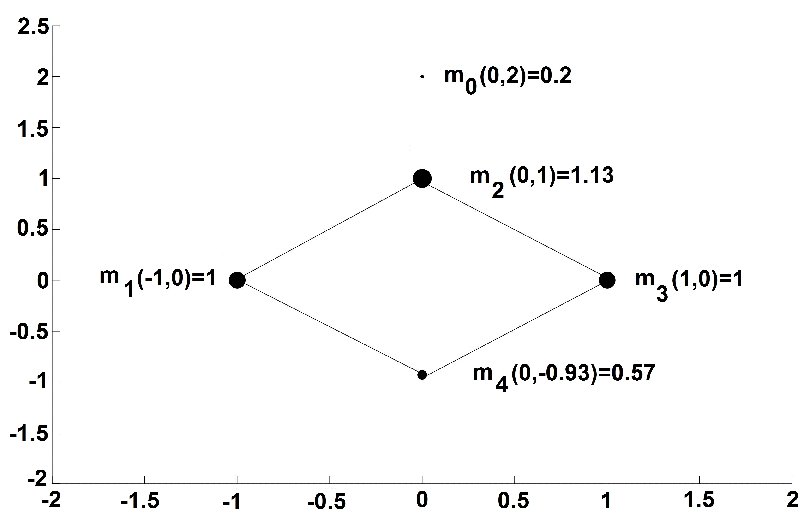}\\
\includegraphics[width=0.45\textwidth]{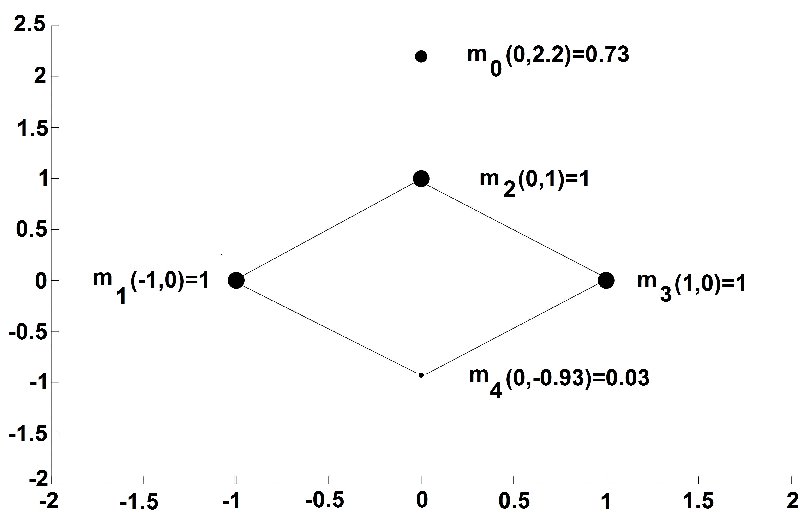}
\caption{Three examples of rhomboidal five-body CC when $w>0$ 
}
\label{Rhombus1}
\end{figure}
\begin{figure}[!htb]
\centering
\includegraphics[width=0.45\textwidth]{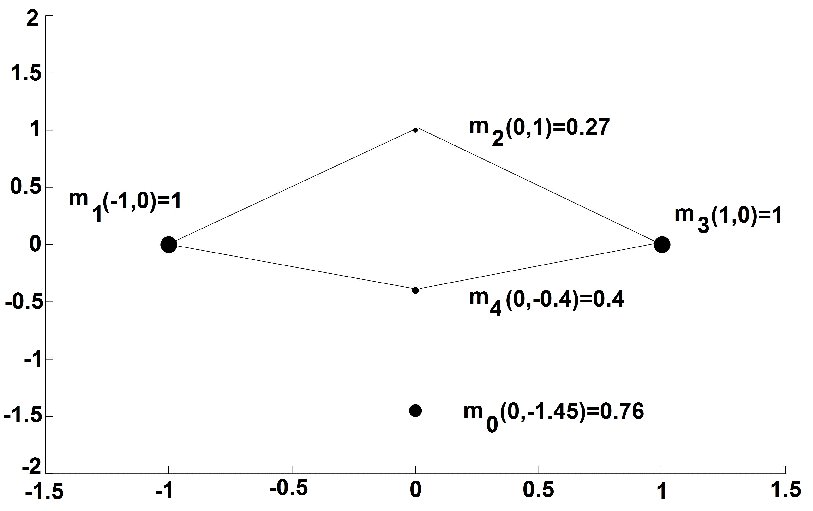}\\
\includegraphics[width=0.45\textwidth]{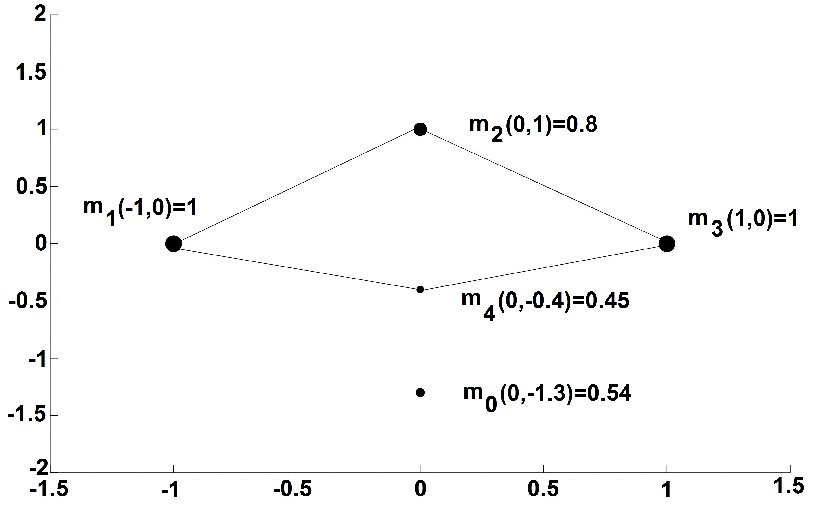}\\
\includegraphics[width=0.45\textwidth]{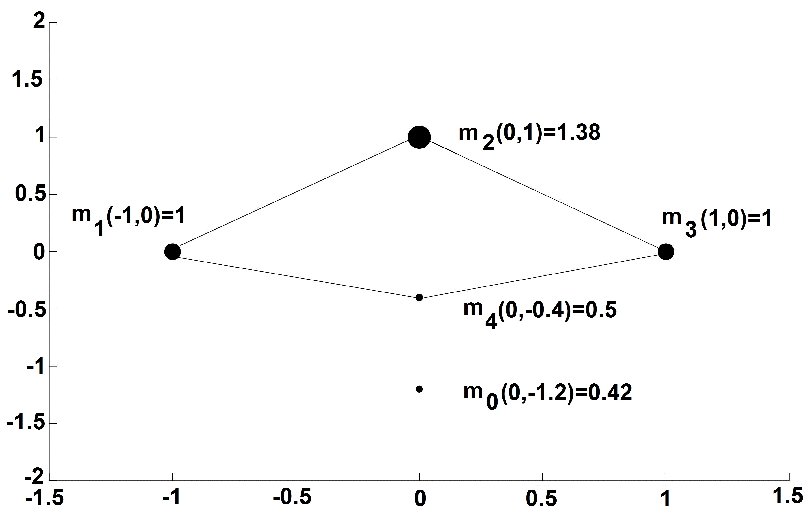}
\caption{Three examples of rhomboidal five-body CC when $w<0$ 
}
\label{Rhombus2}
\end{figure}
The considered variables $t\in (-1,0)$ and $s=1$ in Lemma 1, will give the five-body triangular configurations shown in figure 2.

\begin{corollary}
\label{CorTriangle} Consider $t\in (-1,0)$ in the setup of Theorem
\ref{theorem4}, guaranteeing a triangular five-body arrangement for $w>0$. Then
the configuration $(m_{0}$, $m_{1}$, $m_{2}$, $m_{1}$, $m_{4})$ will
form a central configuration in region:
\begin{equation}
R_{t_{-}}(t,w)=R_{m_{0}}(t,w)\cap R_{m_{2}}(t,w)\cap R_{m_{4}}(t,w).
\end{equation}
\end{corollary}

\begin{proof}
Consider $t<0$ and solve the equation (\ref{14}) in the same way as in the
proof of Theorem \ref{theorem4} to obtain the following regions of central
configurations for the isosceles triangular five-body problem, given in figure
2. We will give a sketch of the proof, and details to the
interested readers (as it follows the same procedure as in Theorem \ref%
{theorem4}). The denominator is always negative when $w>0$. Therefore, the
mass ratios $\mu _{0}$, $\mu _{2}$ and $\mu _{4}$ will be positive when their
respective numerators are negative. Thus, the central configuration regions
where $\mu _{0}$, $\mu _{2}$ and $\mu _{4}$ are respectively positive, are given
below. 
\begin{eqnarray*}
R_{m_{0}} &=&(-1<t<-w\wedge 0<w<1)\vee \\
&\vee&(1<w<1.73\wedge (-1<t<d_{2}(w)\vee \\
&\vee&\frac{0.5}{w}-0.5 w<t<0))\vee 1.5<w<p),
\end{eqnarray*}%
\begin{eqnarray*}
R_{m_{2}} &=&(0<w<1\wedge (-1<t<-w\vee -w<t<0)) \\
&\vee &(-1<t<d_{2}(w)\wedge 1.73205<w<2.41421) \\
&\vee &t>d_{2}(w),
 \end{eqnarray*}%
\begin{eqnarray*}
R_{m_{4}} &=&(0<w<1\wedge -w<t<0)\vee \\
&\vee&(1<w<1.73205\wedge -1<t<0)\vee \\
&\vee & (1.73<w<q)\vee (1.6<w<r),
\end{eqnarray*}%
where
\begin{eqnarray*}
p(t) &=&84.13t^{5}+301.01t^{4}+424.28t^{3}+295.98t^{2}+ \\
&+&103.56t+16.68, \\
q(t) &=&-3.91t^{3}-7.57t^{2}-4.22t+1.22, \\
r(t) &=&305.66t^{3}+52t^{2}+0.25t+1.76.
\end{eqnarray*}%
The central configuration region for the triangular five-body problem is
\begin{equation}
R_{t_{-}}(t,w)=R_{m_{0}}(t,w)\cap R_{m_{2}}(t,w)\cap R_{m_{4}}(t,w).
\end{equation}%
This region is given in figure \ref{Regions-tminus}d alongside $R_{m_{0}}$
(figure \ref{Regions-tminus}a)$,R_{m_{2}}($figure \ref{Regions-tminus}b) and
$R_{m_{4}}$(figure \ref{Regions-tminus}c). The CC region in the figure \ref%
{Regions-tminus} corresponds to the triangular solutions of the five-body problem.
\begin{figure}[!htb]
\centering
\includegraphics[width=0.23\textwidth]{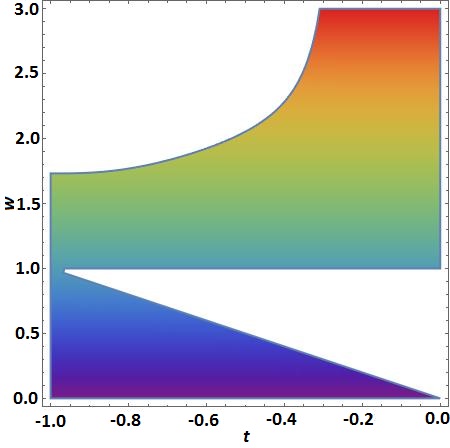}
\includegraphics[width=0.23\textwidth]{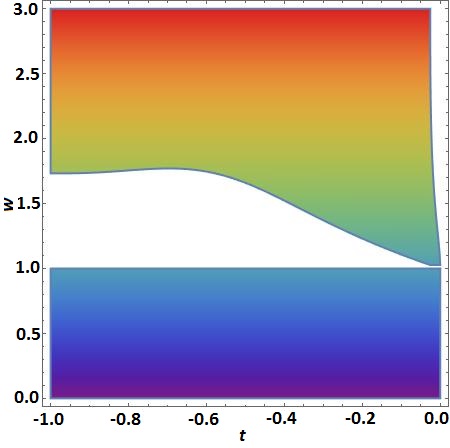}
\includegraphics[width=0.23\textwidth]{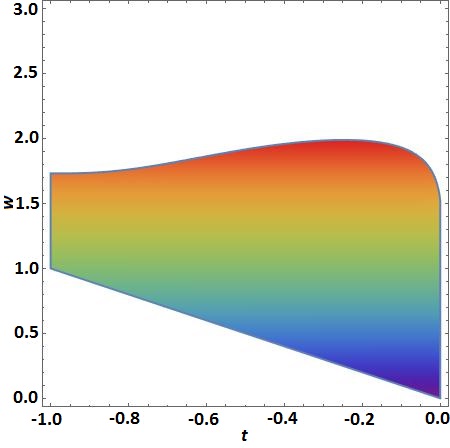}
\includegraphics[width=0.23\textwidth]{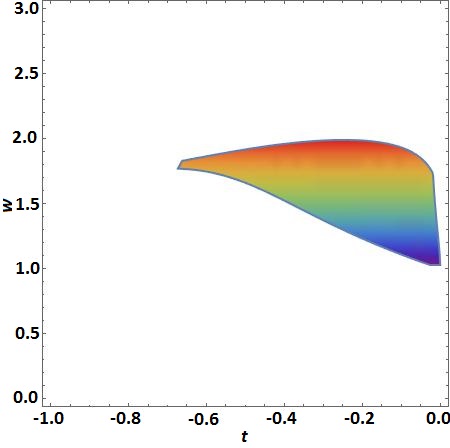}
 \caption{Regions of central configuration in the case of triangular five-body
problem a) $m_{0}>0$, $R_{m_{0}}(t,w)$ b) $m_{2}>0$, $R_{m_{2}}(t,w)$ c) $%
m_{4}>0$, $R_{m_{4}}(t,w)$ d) $m_{i}>0$, $i=0,2,4$, $R_{t_{-}}(t,w)$.}
\label{Regions-tminus}
\end{figure}
\end{proof}

Similarly, by taking $s=\sqrt{3}$, the configuration in figure 2 will become
an equilateral triangle. The CC regions for this equilateral triangle five-body
configuration can be found in the same way as in Theorem \ref{theorem4} or
Corollary \ref{CorTriangle}. Figure  \ref{triangles} shows the physical position of masses in the case of triangular five-body central configurations.
\begin{figure}[!htb]
\centering
\includegraphics[width=0.45\textwidth]{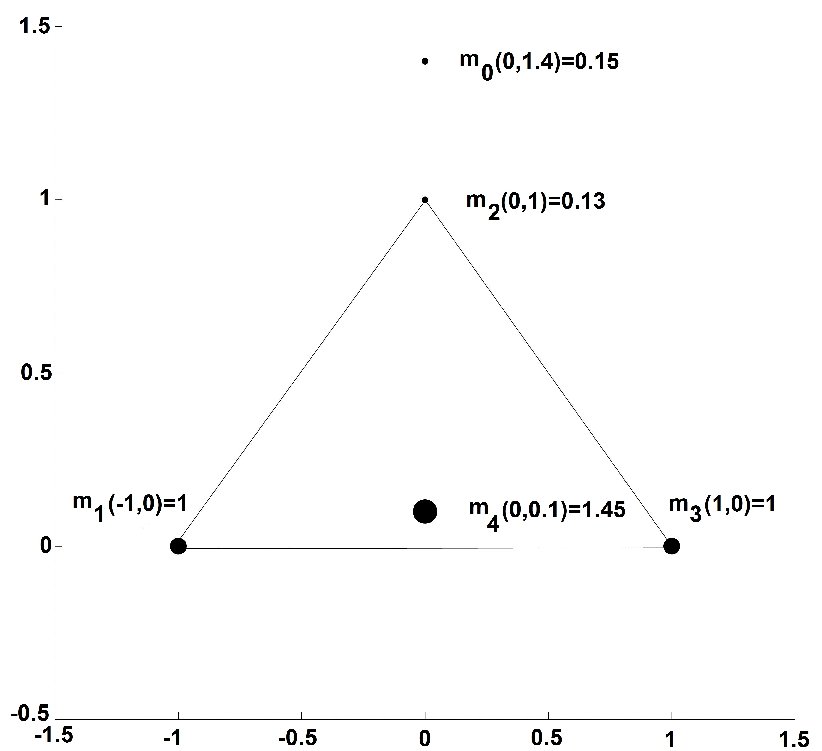}\\
\includegraphics[width=0.45\textwidth]{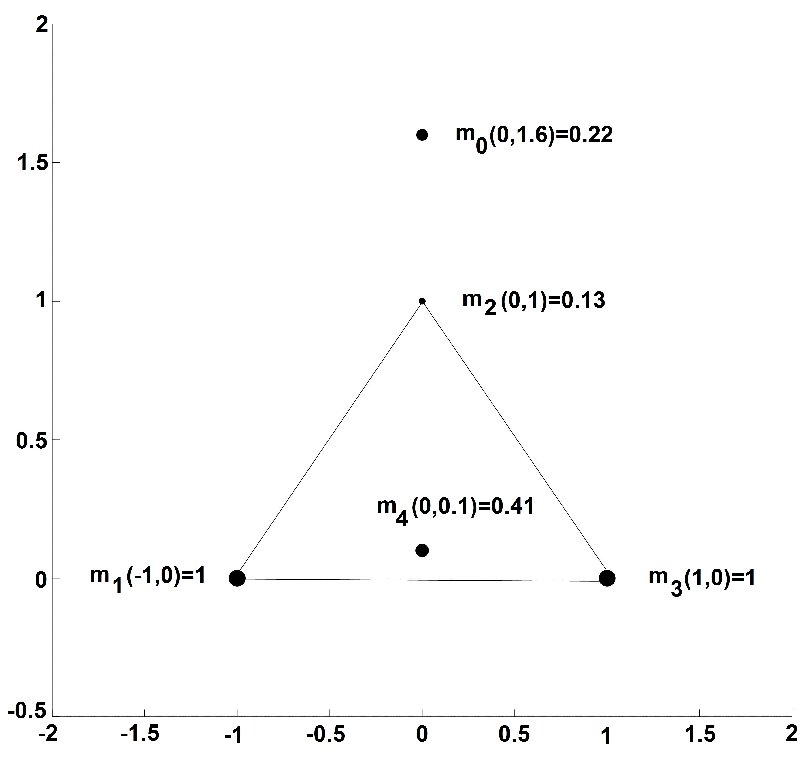}\\
\includegraphics[width=0.45\textwidth]{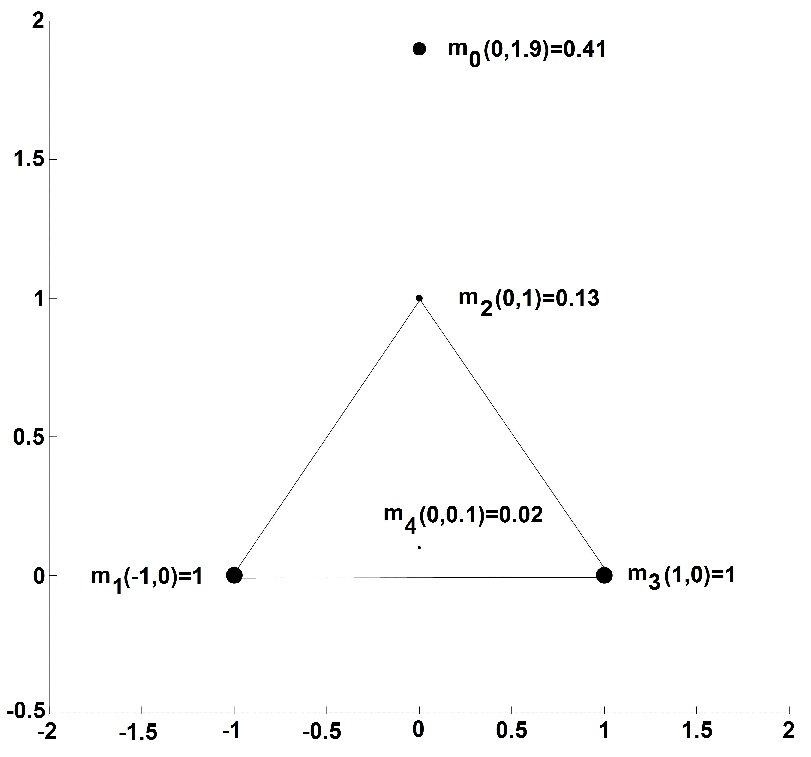}
\caption{Three examples of triangular five-body CC 
}
\label{triangles}
\end{figure}

\subsection*{Note:}

It is impossible to find analytically the CC regions in the full general cases,
i.e. without fixing $s=1$. However, in the current analysis, the majority of the
solutions are obtained. 

\section{Hamiltonian formulation and numerical applications}

In this section, we will introduce the Hamiltonian formalism for the
rhomboidal five-body problem. This five-body problem is introduced in Lemma 1. Then we are going to use
regularized equations and Poincar\'{e} surface of section to investigate the
effect of the increasing value of central mass on the stability of the
proposed five-body problem. Using Poincar\'{e} surface of section we will
study two special cases to identify chaotic regions and quasi-periodic
orbits, with a special focus on the varying central mass. The two cases are:

\begin{enumerate}
\item[\textbf{{a.}}] Four equal masses with a varying central mass.

\item[\textbf{{b.}}] Two pairs of equal masses at the vertices of the
rhombus with a varying central mass.
\end{enumerate}

Two particular cases with only two degrees of freedom are known for the
four-body problem, namely the one-dimensional symmetric four-body problem and
the rhomboidal symmetric four-body problem (see \cite{Lacomba1993}). These
two problems are simpler, but they still have chaotic solutions.
As \cite{diacu2003} showed, symmetries played an essential role in searching for periodic orbits in most of the problems of celestial mechanics \citep{miocbarbosu2003}.
The Hamiltonian
system in two degrees of freedom can be conveniently investigated using the
Poincar\'{e} surface of section. The Poincar\'{e} section plays
an important role in understanding the few-body problem, specifically the
existence of periodic orbits \citep{burgos2013}.\\
Let us consider five point masses $m_{1}(x_{1}=-x_{3}, 0)$, $m_{2}(0, y_{2})$, $m_{3} (x_{3},0)$, $m_{4}(0, -y_{4})$, and  $m_{5}
(0, y_{5})$ in a fixed plane $Oxy$ (see Fig. \ref%
{S1-5}). Let us take $m_{1}=m_{3}$, and the generalized momenta as $p_{i}^{\prime }$, where $i=\overline{1,n}$, to get the Hamiltonian:
\begingroup\makeatletter\def\f@size{8}
\begin{eqnarray}\label{generalHam}
H &=&\sum_{i=1}^{5}\frac{{p^{\prime }}_{i}^{2}}{m_{i}}-\frac{2m_{1}m_{2}}{%
\sqrt{x_{3}^{2}+{y_{2}}^{2}}}-\frac{2m_{1}m_{4}}{\sqrt{x_{3}^{2}+{y_{4}}^{2}}%
}-\frac{2m_{1}m_{5}}{\sqrt{x_{3}^{2}+{y_{5}}^{2}}}- \nonumber\\
&-&\frac{m_{1}^{2}}{2x_{3}}-\frac{m_{2}m_{4}}{\sqrt{(y_{4}+y_{2})^{2}}}-%
\frac{m_{2}m_{5}}{\sqrt{(y_{5}-y_{2})^{2}}}-\frac{m_{4}m_{5}}{\sqrt{%
(y_{4}+y_{5})^{2}}}.
\end{eqnarray}
\endgroup
Consider $m_{2}=m_{4}$ to be symmetric on the axis $Oy$, ($y_{2}=y_{4}$), and the mass $m_{5}=m_{0}$ stationary at the origin ($y_{5}=0$%
). The Hamiltonian then becomes: \begingroup\makeatletter%
\begin{eqnarray}
H &=&\frac{{p^{\prime }}_{2}^{2}}{2m_{2}}+\frac{{p^{\prime }}_{3}^{2}}{2m_{1}%
}-\frac{2m_{1}m_{2}}{\sqrt{x_{3}^{2}+{y_{2}}^{2}}}-  \nonumber \\
&-&\frac{m_{1}(m_{1}+4m_{0})}{4x_{3}}-\frac{m_{2}(m_{2}+4m_{0})}{4y_{2}}.
\end{eqnarray}%
\endgroup
Represent the generalized coordinates as $x_{3}:=q_{1}$, and $%
y_{2}:=q_{2}$, and the generalized momenta as $p_{3}^{\prime }:=p_{1}$, and $%
p_{4}^{\prime }:=p_{2}$, then the corresponding Hamiltonian takes the form %
\begingroup\makeatletter
\begin{eqnarray}
H &=&\frac{p_{1}^{2}}{2m_{1}}+\frac{p_{2}^{2}}{2m_{2}}-\frac{2m_{1}m_{2}}{%
\sqrt{q_{1}^{2}+{q_{2}}^{2}}}-  \nonumber \\
&-&\frac{m_{1}(m_{1}+4m_{0})}{4q_{1}}-\frac{m_{2}(m_{2}+4m_{0})}{4q_{2}},
\end{eqnarray}%
\endgroup The corresponding canonical equations of motion are
\begin{eqnarray}
q_{1} &=&\frac{p_{1}}{m_{1}},  \label{eqofmotion} \\
q_{2} &=&\frac{p_{2}}{m_{2}},  \nonumber \\
p_{1} &=&-\frac{2m_{1}m_{2}}{(q_{1}^{2}+{q_{2}}^{2})^{3/2}}q_{1}-\frac{%
m_{1}(m_{1}+4m_{0})}{2q_{1}^{2}},  \nonumber \\
p_{2} &=&-\frac{2m_{1}m_{2}}{(q_{1}^{2}+{q_{2}}^{2})^{3/2}}q_{2}-\frac{%
m_{2}(m_{2}+4m_{0})}{2q_{2}^{2}}.  \nonumber
\end{eqnarray}%
Since we have singularities in the equations of motion, we use a
regularization technique, called the double Levi-Civita transformation to
regularize the Hamiltonian (see \cite{szucs1}).

Introduce the fictitious time $\tau $ as
\begin{equation}
\frac{dt}{d\tau }=q_{1}q_{2},
\end{equation}%
and the transformed coordinates as
\begin{equation}
q_{i}=Q_{i}^{2},i=1,2.  \label{LCcoord}
\end{equation}

In order to achieve regularization we extend the coordinate transformation
to a canonical transformation by introducing the generating function $W$
(see \cite{csillik2003}) as
\begin{equation}
W=p_{1}Q_{1}^{2}+p_{2}Q_{2}^{2}.
\end{equation}%
From the generating function $P_{j}=\sum\limits_{i=1}^{2}p_{i}\frac{\partial
q_{i}}{\partial Q_{j}}$, $j=1,2$, we obtain the transformed  momenta
\begin{equation}
p_{i}=\frac{P_{i}}{2Q_{i}},i=1,2,  \label{LCmomenta}
\end{equation}%
where $P_{i}$ are the new momenta.

The canonical symplectic form of equations (\ref{eqofmotion}) will be
preserved, if the new Hamiltonian is
\begin{equation}
\overline{H}=H-h_{0},
\end{equation}%
where $h_{0}$ is the constant of total energy.

Therefore, the transformed Hamiltonian  is given by \begingroup\makeatletter
\begin{eqnarray}
\overline{H} &=&\frac{1}{8}\left( \frac{P_{1}^{2}Q_{2}^{2}}{m_{1}}+\frac{%
P_{2}^{2}Q_{1}^{2}}{m_{2}}\right) -\frac{1}{4}(m_{1}(4m_{0}+m_{1})Q_{2}^{2}
\nonumber \\
&+&m_{2}(4m_{0}+m_{2})Q_{1}^{2})-\frac{2m_{1}m_{2}Q_{1}^{2}Q_{2}^{2}}{\sqrt{%
Q_{1}^{4}+Q_{2}^{4}}}-h_{0}Q_{1}^{2}Q_{2}^{2},  \nonumber
\end{eqnarray}%
\endgroup and the transformed equations of motion are
\begin{eqnarray}
\frac{dQ_{i}}{d\tau } &=&\frac{\partial \overline{H}}{\partial P_{i}},
\label{neweqofmotion} \\
\frac{dP_{i}}{d\tau } &=&-\frac{\partial \overline{H}}{\partial Q_{i}},\quad
i=1,2.  \nonumber
\end{eqnarray}%
To construct the Poincar\'{e} surface of section and find the corresponding
periodic orbits (see \cite{cheb1996}), we plot the motion from the 4D phase
space ($Q_{1}$, $Q_{2}$, $P_{1}$, $P_{2}$) in a "cut" plane ($Q_{1}=0$, $%
Q_{2}$, $P_{1}>0$, $P_{2}$). Since $\overline{H}$ is conserved,
any point in this surface of section will uniquely define the orbit.

To investigate the effect of central mass on the existence of quasi-periodic
orbits, consider $m_{1}=m_{2}=$$m_{3}=$$m_{4}=1$ and

\begin{itemize}
\item[(i.)] $m_{0}=1$,
\item[(ii.)] $m_{0}=3$,
\item[(iii.)] $m_0=5$.
\end{itemize}

Consider $Q_{1}(0)=0$ and $P_{1}(0)>0$, to investigate the motion in the
$(Q_{2},P_{2})$ plane, where the corresponding constant total energy $h_{0}$
have the values $-0.46$, $-0.9,$ and $-1.3$. Figures \ref{poin111}, \ref%
{poinM311}, and \ref{poinM511} show the results for the gradually increased
central mass. In addition, we also give a representative periodic orbit in
each case.
\begin{figure}[!htb]
\includegraphics[width=0.45\textwidth]{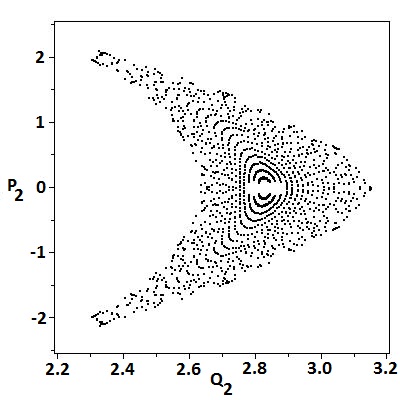}\\ %
\includegraphics[width=0.45\textwidth]{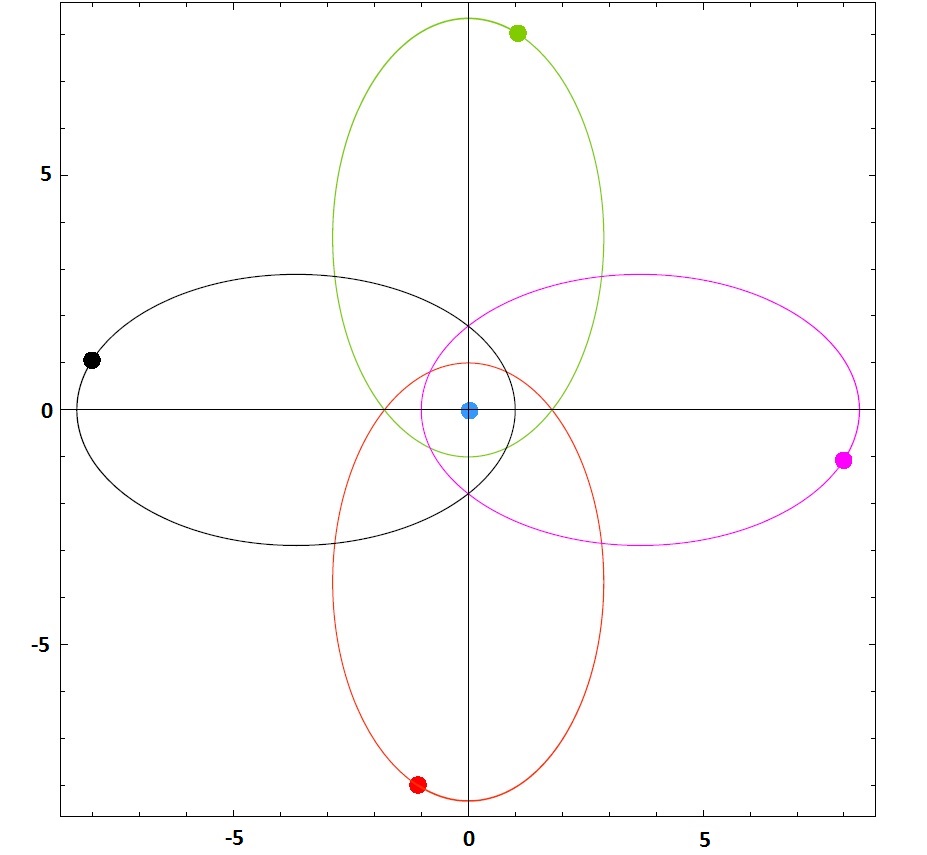}
\caption{a) Poincar\'{e} section, when the central mass is $1$. b) A
representative periodic orbit corresponding to the centre.}
\label{poin111}
\end{figure}

\begin{figure}[!htb]
\includegraphics[width=0.45\textwidth]{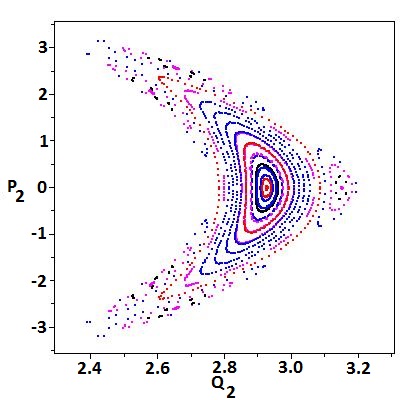}\\ %
\includegraphics[width=0.45\textwidth]{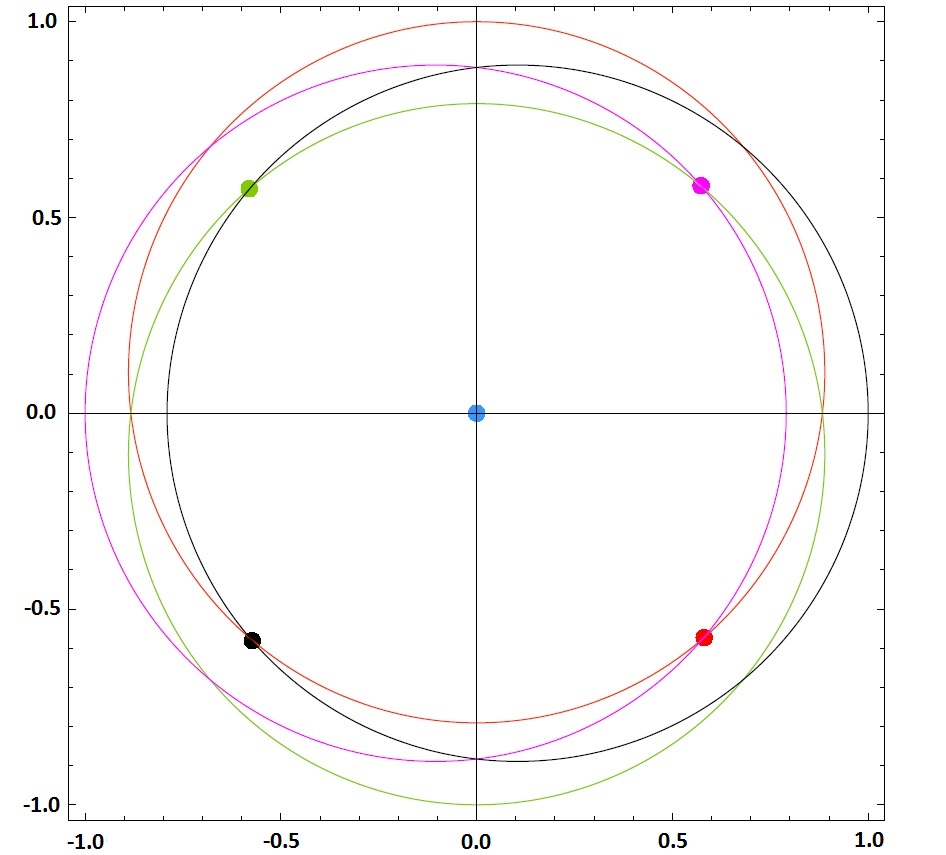}
\caption{Poincar\'{e} section, when the central mass is $3$. b) A
representative periodic orbit corresponding to the centre}
\label{poinM311}
\end{figure}

\begin{figure}[!htb]
\includegraphics[width=0.45\textwidth]{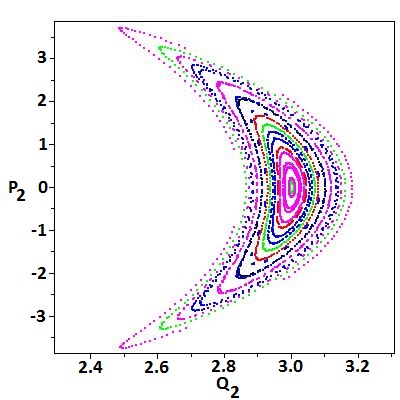}\\ %
\includegraphics[width=0.45\textwidth]{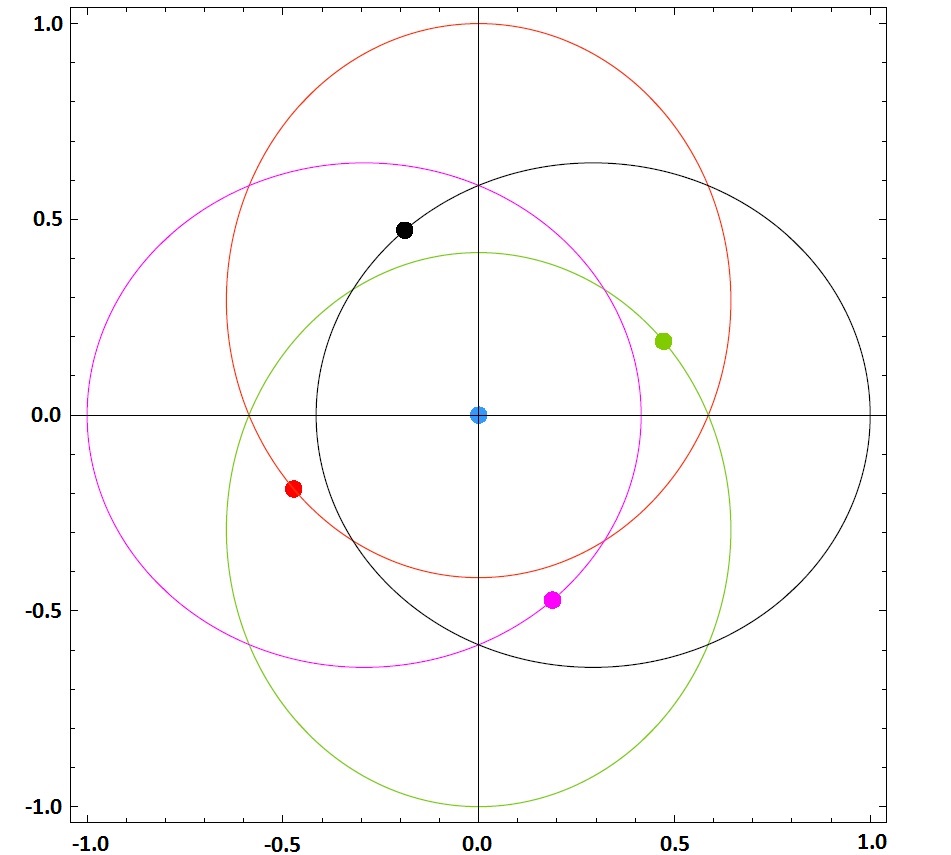}
\caption{a) Poincar\'{e} section, when the central mass is $5$. b) A
representative periodic orbit corresponding to the centre.}
\label{poinM511}
\end{figure}
In the equal mass case, there are a couple of quasi-periodic orbits at and
around $P_{2}=0$ and $Q_{2}=2.83$ (Fig. \ref{poin111}). The rest of the
points in figure \ref{poin111} are indicative of chaotic behavior. In figure %
\ref{poinM311}, the central mass is increased to $m_{0}=3$. It is clear from
figure \ref{poinM311} that the region with periodic orbits has begun to
increase in size. Although the outer region is still chaotic but
in inner region appears some regular structures. This is indicative of some
underlying structure in the dynamics. At $m_{0}=5$, five times bigger than
the outer masses, we can see higher number of quasi-periodic orbits around $%
Q_{2}=3$. The chaotic behavior which was very clear in the equal mass case
has completely disappeared. This indicates that the increase in central mass
plays a stabilizing role. Similar behavior was reported by \cite{Shoaib2008}
for the hierarchical stability of the Caledonian symmetric five-body problem
(CS5BP).

To complete the analysis given above, we discuss two more examples with
two pairs of equal masses at the vertices of the rhombus, and other example with one pairs of equal masses on the $x-$axis and three different masses on the $y-$axis (see Fig. \ref{Rhombus1}, \ref{Rhombus2} and \ref{triangles}).
Consider
\begin{itemize}
\item[a.] $m_{1}=1$, $m_{2}=2$, and the central mass $m_{0}=0.05$ ($%
h_{0}=-0.9$);

\item[b.] $m_{1}=1$, $m_{2}=1.1$, and the central mass $m_{0}=3$ ($%
h_{0}=-0.9$).


\item[c.] $m_{1}=m_{3}=1$, $m_{2}=0.27$, $m_4=0.4$, $m_0=0.76$
($h_{0}=-0.9$).

\end{itemize}
\begin{figure}[!htb]
\centering
\includegraphics[width=0.45\textwidth]{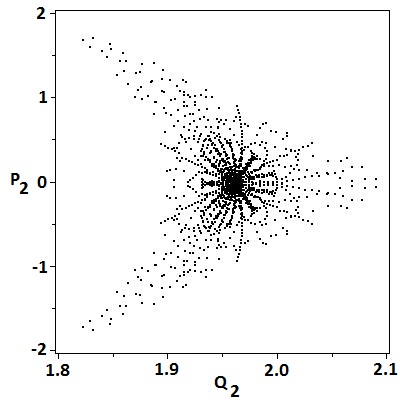}
\caption{Poincar\'{e} section in case a.}
\label{poinM0051209}
\end{figure}
\begin{figure}[!htb]
\centering
\includegraphics[width=0.45\textwidth]{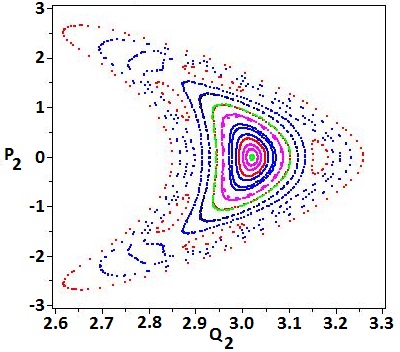}
\caption{Cross section at $Q_{1}=0$ in case b., showing the quasi-periodic orbits.}
\label{poinM3111}
\end{figure}

As in case a., the surfaces in figures \ref{poinM0051209} and \ref%
{poinM3111} show various types of orbits, including both circle-like
quasi-periodic orbits and island orbits. Note the deformation in some of the
circular orbits, indicating the presence of nearby island orbits. When the
central body is small, the surface of section presents very few
quasi-periodic orbits surrounded by chaotic region. As before, the increase
in central mass completely changes the dynamics of the problem.
A direct comparison of figures \ref{poin111} with figures \ref{poinM3111}
reveal the obvious effect of changing central mass on the stability and
existence of quasi-periodic orbits in the rhomboidal five-body problem
introduced in section 1.
\begin{figure}[!htb]
\centering
\includegraphics[width=0.45\textwidth]{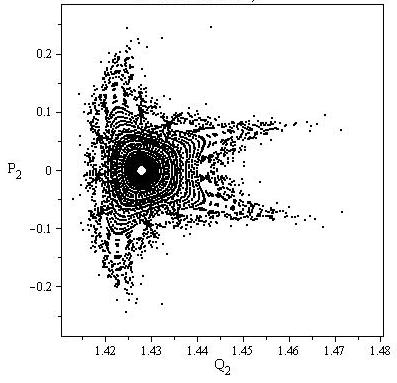}
\caption{Poincar\'{e} section in case c.}
\label{last}
\end{figure}

In case c. (see Fig. \ref{Rhombus2}a) three non-equal masses are distributed on the $y-$axis. We used the general form of the Hamiltonian (equation \ref{generalHam}), and the notations: $x_3:=q_1$, $y_2:=q_2$, $m_5$=$m_0$.  $y_4$ and $y_5$ are functions of $q_2$. As in previous cases, we constructed the corresponding new Hamiltonian (this transformation is trivial) and took $y_4=q_2-1.4$ and $y_5=q_2-2.45$. In this case we obtain
the Poincar\'{e} surface of section shown in figure \ref{last}. The inner region of figure \ref{last} points to the existence of regular structure and some quasi periodic orbits. The outer region exhibits chaotic behaviour.


\section{Conclusions}

In this paper we studied the central configurations of rhomboidal and triangular four- and five-body problems. (The configuration has only one symmetry with $m_{1}=m_{3}\neq m_{2}\neq m_{4}\neq m_{0}$). First, Dziobek's equations are derived for this particular set up. Then, we derive
regions of central configuration for various special cases of the general
problem. We prove that no central configurations are possible when two pairs of masses are placed at the vertices of a
rhombus, (which is symmetric both about the $x-$axis and the $y-$axis) and the fifth mass is placed anywhere on the axis of symmetry, except the origin. Using
analytical techniques, regions of central configurations are derived for all
other cases, including rhomboidal four- and five-body problems,
isosceles and equilateral triangular five-body problems. To complement the
analytical results, these regions are also explored numerically. Equations of
motion are regularized using the double Levi-Civita transformations, and from
the regularized equations we identify regions with quasi-periodic orbits. Moreover, we investigate the chaotic behavior in the phase space by means of the  Poincar\'{e} sections. We show that if we increase the central mass, this has a
stabilizing effect on the rhomboidal five-body configuration.

\section*{Appendix}

\subsection*{\textbf{Derivation of $R_{N_{\mu _{0}}}$}: }

The numerator of $\mu _{0}$, $N_{\mu _{0}}$ is positive when ${A}_{{1}}{C}_{{%
2}}{C}_{{3}}$ $>0$ and ${B}_{{1}}{B}_{{3}}{C}_{{2}}+{B}_{{2}}{C}_{{1}}{C}_{{3}%
}<0$ or when both factors have opposite sign and the positive factor is
greater than the absolute value of the negative factor. All these
possibilities are listed below with the corresponding regions, where $N_{\mu
_{0}}>0$

\begin{enumerate}
\item ${A}_{{1}}{C}_{{2}}{C}_{{3}}>0$ and ${B}_{{1}}{B}_{{3}}{C}_{{2}}+{B}_{{%
2}}{C}_{{1}}{C}_{{3}}<0$: Using numerical approximation techniques it is
obtained that $N_{\mu _{0}}>0$ in the following region:%
\begin{equation}
R_{cN_{\mu _{0}}}(t,w)=R_{aN_{\mu _{0}}}(t,w)\cup R_{bN_{\mu _{0}}}(t,w),
\end{equation}%
where \begingroup\makeatletter%
\begin{eqnarray*}
R_{aN_{\mu _{0}}}(t,w) &=&\{(t,w)|(0<t<0.41\wedge 0<w<1) \\
&\vee &(0.41<t<1\wedge w>2.41)\vee (-1.73 \\
&<&w\leq -1.5\wedge 0.75\cdot d_{2}(w)<t<1)\}, \\
R_{bN_{\mu _{0}}}(t,w) &=&(w<-1.8\wedge 0<t<0.23)\vee (-1.43 \\
&<&w\leq -1\wedge 0.56<t<1)\vee (-1<w \\
&<&-0.41\wedge 0.41<t<-w)\vee (2<w \\
&<&2.41\wedge 0.41<t<1)\}.
\end{eqnarray*}%
\endgroup

\item ${A}_{{1}}{C}_{{2}}{C}_{{3}}>0$ and ${B}_{{1}}{B}_{{3}}{C}_{{2}}+{B}_{{%
2}}{C}_{{1}}{C}_{{3}}>0$: It is numerically confirmed that ${B}_{{1}}{B}_{{3}%
}{C}_{{2}}+{B}_{{2}}{C}_{{1}}{C}_{{3}}>{A}_{{1}}{C}_{{2}}{C}_{{3}}$ when ${A}%
_{{1}}{C}_{{2}}{C}_{{3}}>0.$ Hence, $N_{\mu _{0}}<0$ for all $(t,w)$ such
that ${A}_{{1}}{C}_{{2}}{C}_{{3}}>0.$

\item ${A}_{{1}}{C}_{{2}}{C}_{{3}}<0$ and ${B}_{{1}}{B}_{{3}}{C}_{{2}}+{B}_{{%
2}}{C}_{{1}}{C}_{{3}}<0$. In this case using numerical approximation
techniques we find that $N_{\mu _{0}}$ is positive in the following region: %
\begingroup\makeatletter%
\begin{eqnarray*}
R_{dN_{\mu _{0}}}(t,w) &=&\{(t,w)|(-1.73<w<-1.5\wedge \\
&\wedge&0<t<0.25)\vee (-1.5\leq w<-1\wedge  \\
&\wedge&0<t<0.2\cdot d_{2}(w))\vee (0<t<0.41\\
&\wedge& -t<w <0) \vee  (0.41<t<1 \wedge\\
&\wedge& 0<w<1)\}.
\end{eqnarray*}%
\endgroup
\end{enumerate}

Hence, $N_{\mu _{0}}>0$ in (see Fig. \ref{Regions-m0}b):
\begin{equation}
R_{N_{\mu _{0}}}=R_{cN_{\mu _{0}}}(t,w)\cup R_{dN_{\mu _{0}}}(t,w).
\end{equation}%

\subsection*{\textbf{Derivation of $R_{N_{\mu _{2}}}$}: }

 The numerator of $\mu_2$, $N_{\mu _{2}}$, is always
positive when $A_{3}B_{2}C_{1}>0$ and $A_{1}A_{3}C_{2}+A_{2}C_{1}C_{3}<0.$
It is not necessarily positive when $A_{3}B_{2}C_{1}>0$ and $%
A_{1}A_{3}C_{2}+A_{2}C_{1}C_{3}<0$ or $A_{3}B_{2}C_{1}<0$ and $%
A_{1}A_{3}C_{2}+A_{2}C_{1}C_{3}<0.$ It might be positive in a segment of
this region or might not be positive at all. We will list all these
possibilities below with the corresponding regions, where $N_{\mu _{2}}>0.$

\begin{enumerate}
\item $A_{3}B_{2}C_{1}>0$ and $A_{1}A_{3}C_{2}+A_{2}C_{1}B_{3}<0$. Using
numerical approximation techniques we show that $N_{\mu _{2}}>0$ in%
\begin{equation}
R_{aN_{\mu _{2}}}(t,w)=aaN_{\mu _{2}}(t,w)\cup baN_{\mu _{2}}(t,w)\cup
caN_{\mu _{2}}(t,w),
\end{equation}%
where \begingroup\makeatletter
\begin{eqnarray*}
aaN_{\mu _{2}}(t,w) &=&\{(t,w)|(0<w\leq 0.41\wedge d_{1}(w)< \\
&<& t<1) \vee (0.41<w<0.58 \wedge \\
&\wedge &d_1(w)<t<0.5\cdot d_{2}(w)) \\
&\vee &(w>\sqrt{3}\wedge 0<t<d_{1}(w))\},
\end{eqnarray*}
\begin{eqnarray*}
baN_{\mu _{2}}(t,w) &=&\{(t,w)|(w\leq -2.41\wedge 0.463-\\
&-& 0.079\cdot w <t<1)\vee (-2.41<w<\\
&<& -1.73\wedge 0.463-0.079w < \\
&<& t < 0.5 d_{2}(w))\},\\
caN_{\mu _{2}}(t,w) &=&\{(t,w)|1<w<1.73\wedge\\
&\wedge& 0<t<d_{1}(w)\}.
\end{eqnarray*}%
\endgroup

\item $A_{3}B_{2}C_{1}>0$ and $A_{1}A_{3}C_{2}+A_{2}C_{1}B_{3}>0$. Using
numerical approximation techniques we show that $N_{\mu _{2}}>0$ in%
\begin{equation}
R_{bN_{\mu _{2}}}(t,w)=abN_{\mu _{2}}(t,w)\cup bbN_{\mu _{2}}(t,w)\cup
cbN_{\mu _{2}}(t,w),
\end{equation}%
where%
\begin{eqnarray*}
abN_{\mu _{2}}(t,w) &=&\{(t,w)|-1.73<w<-1 \wedge\\
&\wedge &0<t<0.5\cdot d_{2}(w)\},
\end{eqnarray*}
\begin{eqnarray*}
bbN_{\mu _{2}}(t,w) &=&\{(t,w)|1<w<1.73 \wedge\\
&\wedge &0<t<d_{1}(w)\},
\end{eqnarray*}
\begin{eqnarray*}
cbN_{\mu _{2}}(t,w) &=&\{(t,w)|(w\leq -2.41\wedge \\
&\wedge& 0<t<1)\vee (-2.41<w<-1.73 \\
&\wedge& 0<t<0.5\cdot d_{2}(w))\}.
\end{eqnarray*}

\item $A_{3}B_{2}C_{1}<0$ and $A_{1}A_{3}C_{2}+A_{2}C_{1}B_{3}<0$. Using
numerical approximation techniques we show that $N_{\mu _{2}}>0$ in%
\begin{equation}
R_{cN_{\mu _{2}}}(t,w)=acN_{\mu _{2}}(t,w)\cup bcN_{\mu _{2}}(t,w),
\end{equation}%
where\begingroup\makeatletter
\begin{eqnarray*}
acN_{\mu _{2}}(t,w) &=&\{(t,w)|(-2.4<w<-2\wedge  \\
&\wedge& 0.5\cdot d_{2}(w)<t<1)\vee  \\
&\wedge& (1<w<1.73\wedge d_{1}(w)<t<1)\}, \\
bcN_{\mu _{2}}(t,w) &=&\{(t,w)|(1.31\cdot t^{2}-2.012\cdot t+0.9<\\
&<&w<0.58\wedge 0<t<d_{1}(w))\vee  \\
&\vee& (1.31\cdot t^{2}-2.012 \cdot t+0.9\leq w<1 \\
&\wedge& 0<t<0.5 \cdot d_{2}(w))\vee  \\
&\vee& (w>1.732\wedge d_{1}(w)<t<1)\}.
\end{eqnarray*}%
\endgroup Therefore, $N_{\mu _{2}}(t,w)>0$ in region%
\begin{equation}
R_{N_{\mu _{2}}}=R_{aN_{\mu _{2}}}(t,w)\cup R_{bN_{\mu _{2}}}(t,w)\cup
R_{cN_{\mu _{2}}}(t,w).
\end{equation}%
\end{enumerate}

\subsection*{\textbf{Derivation of $R_{N_{\mu _{4}}}$}: }

$N_{\mu _{4}}$ is always
positive when $A_{2}B_{1}C_{3}>0$ and $A_{3}B_{1}B_{2}+A_{1}A_{2}C_{3}<0.$
It is not necessarily positive when $A_{2}B_{1}C_{3}>0$ and $%
A_{3}B_{1}B_{2}+A_{1}A_{2}C_{3}<0$ or when both the factors are negative. It
might be positive in a segment of this region or might not be positive at
all. We will list all these possibilities below with the corresponding
regions, where $N_{\mu _{4}}>0$.

\begin{enumerate}
\item When $A_{2}B_{1}C_{3}>0,$ and $A_{3}B_{1}B_{2}+A_{1}A_{2}C_{3}<0$, $%
N_{\mu _{4}}>0$ in region\begingroup\makeatletter
\begin{eqnarray*}
R_{aN_{\mu _{4}}} &=&\{(t,w)|(1.73<w<2.4\wedge 0.<t<d_{1})\vee \\
&\vee& (-1<w<-0.41\wedge -w<t<1)\vee \\
&\vee& (0<w<1 \wedge d_{1}<t<1) \vee \\
&\vee& (w<-377.64t^{4}+677.58t^{3}-458.06t^{2}\\
&+&136.9t-17.79\wedge 0.25<t<0.6) \vee\\
&\vee &(-1<w<-0.41\wedge 0<t<0.05)\vee \\
&\vee& (1.73<w<2.41\wedge d_{1}<t<0.41)\}.
\end{eqnarray*}%
\endgroup

\item When $A_{2}B_{1}C_{3}>0,$ and $A_{3}B_{1}B_{2}+A_{1}A_{2}C_{3}>0$, $%
N_{\mu _{4}}>0$ in region \begingroup\makeatletter
\begin{eqnarray*}
R_{bN_{\mu _{4}}} &=&\{(t,w)|(0.1<w<1\wedge 0.41<t<d_{1}) \\
&\vee &(1<w<1.73\wedge d_{1}<t<0.41) \\
&\vee &(1.73<w<2.41\wedge 0.41<t<1) \\
&\vee &(-1<w<-0.41\wedge -w<t<1) \\
&\vee &(0<w<1\wedge d_{1}<t<1) \\
&\vee &(1<w<1.73\wedge 0<t<d_{1}) \\
&\vee &(1.31t^{2}-2t+0.9<w<1\wedge\\
&\wedge& 0<t<0.41) \vee \\
&\vee &(1<w<1.73\wedge 0.41<t<1) \\
&\vee &(1.73<w<2.41\wedge d_{1}<t<0.41)\}.
\end{eqnarray*}%
\endgroup

\item When $A_{2}B_{1}C_{3}<0,$ and $A_{3}B_{1}B_{2}+A_{1}A_{2}C_{3}<0$, $%
N_{\mu _{4}}>0$ in region \begingroup\makeatletter
\begin{eqnarray*}
R_{cN_{\mu _{4}}} &=&\{(t,w)|(-0.41<w<0\wedge 0<t<-w)\vee  \\
&\vee&(w<-1.3t^{2}-1.9t+3.44\wedge \\
&\wedge& d_{1}<t<0.41) \\
&\vee &(w>2.41\wedge 0<t<d_{1})\}.
\end{eqnarray*}%
\endgroup
Therefore, $N_{\mu _{4}}(t,w)>0$ in  (see Fig. \ref{Regions-m4}a):
\begin{equation}
R_{N_{\mu _{4}}}=R_{aN_{\mu _{4}}}(t,w)\cup R_{bN_{\mu _{4}}}(t,w)\cup
R_{cN_{\mu _{4}}}(t,w).
\end{equation}%
\end{enumerate}

\acknowledgments
We thank the editor and the anonymous reviewer for their constructive comments, which helped us improve the presentation of the manuscript. We also thank Professor Mihail Barbosu for correcting and improving the language of the manuscript. I. Sz\"{u}cs-Csillik is partially supported by a grant of the Romanian Ministry of National Education and Scientific Research, RDI Programme for Space Technology and Advanced Research - STAR, project number 513, 118/14.11.2016.


\end{document}